\documentclass[a4paper,11pt]{scrartcl}

\usepackage{tikz}
\usetikzlibrary{calc,shapes, arrows}

\usepackage[color=blue!50]{todonotes}

\usepackage[utf8]{inputenc}
\usepackage[T1]{fontenc}

\usepackage{amsmath,amsthm,amssymb}
\usepackage{mathtools}

\usepackage[linesnumbered, vlined, ruled]{algorithm2e}
\usepackage{lineno}

\usepackage[bottom=1in,top=1in,left=1in,right=1in]{geometry}
\RedeclareSectionCommand[indent=0pt]{subparagraph}

\usepackage{hyperref}
\usepackage{xspace}
\usepackage{thm-restate}

\newtheorem{theorem}{Theorem}
\newtheorem{definition}[theorem]{Definition}
\newtheorem{lemma}[theorem]{Lemma}

\newtheorem{observation}[theorem]{Observation}
\newtheorem*{technique-no-number}{Technique}
\newtheorem{corollary}[theorem]{Corollary}

\newcommand{\kk}{\texorpdfstring{$k$}{k}}
\newcommand{\RR}{\mathbb{R}}

\DeclareMathOperator*{\argmin}{arg\,min}

\DeclareMathOperator{\cost}{cost}
\def\coloneqq{\mathrel{\mathop:}=}

\title{Achieving anonymity via weak lower bound constraints for $k$-median and $k$-means }

\author{Anna Arutyunova\thanks{University of Bonn, Germany, \href{mailto:arutyunova@informatik.uni-bonn.de}{arutyunova@informatik.uni-bonn.de}} \and Melanie Schmidt\thanks{University of Cologne, Germany, \href{mailto:mschmidt@cs.uni-koeln.de}{mschmidt@cs.uni-koeln.de}}}

\begin{document}
	
	\maketitle

	\begin{abstract}
		We study $k$-clustering problems with lower bounds, including $k$-median and $k$-means clustering with lower bounds. In addition to the point set $P$ and the number of centers $k$, a $k$-clustering problem with (uniform) lower bounds gets a number $B$. The solution space is restricted to clusterings where every cluster has at least $B$ points.
		We demonstrate how to approximate $k$-median with lower bounds via a reduction to facility location with lower bounds, for which $O(1)$-approximation algorithms are known.
		
		Then we propose a new constrained clustering problem with lower bounds where we allow points to be assigned \emph{multiple times} (to different centers). This means that for every point, the clustering specifies a set of centers to which it is assigned. We call this \emph{clustering with weak lower bounds}. We give a $(6.5+\epsilon)$-approximation for $k$-median clustering with weak lower bounds and an $O(1)$-approximation for $k$-means with weak lower bounds. 
		
		We conclude by showing that at a constant increase in the approximation factor, we can restrict the number of assignments of every point to $2$ (or, if we allow fractional assignments, to $1+\epsilon$). This also leads to the first bicritera approximation algorithm for $k$-means with (standard) lower bounds where bicriteria is interpreted in the sense that the lower bounds are violated by a constant factor.
		
		All algorithms in this paper run in time that is polynomial in $n$ and $k$ (and $d$ for the Euclidean variants considered).
	\end{abstract}
	
	\newpage
	\section{Introduction}
	We study $k$-clustering problems with lower bound constraints. 
	Imagine the following approach to publish a reduced version of a large data set: Partition the data into clusters of similar objects, then replace every cluster by one (weighted) point that represents it best. Publish these weighted representatives. 
	For example, it is a fairly natural approach for data that can be modeled as vectors from $\mathbb{R}^d$ to replace a data set by a set of mean vectors, where every mean vector represents a cluster. When representing a cluster by one point, the mean vector minimizes the squared error of the representation. This is a common use case of $k$-means clustering. 
	
	In this paper, we ask the following: If we want to publish the representatives, it would be very convenient if the clusters were of sufficient size to ensure a certain level of anonymity of the individual data points that they represent. Can we achieve this, say, in the case of $k$-means clustering or for the related $k$-median problem?
	
	Using clustering with lower bounds on the cluster sizes to achieve anonymity is an idea posed by Aggarwal et al.~\cite{APFTKKZ10}. They introduce it in the setting of radii-based clustering, and define the \emph{$r$-gather problem}: Given a set of points $P$ from a metric space, find a clustering and centers for the clusters such that the maximum distance between a point and its center is minimized and such that every cluster has at least $r$ points. They also define the {$(k,r)$-center problem} which is the same problem as the $r$-gather problem except that the number of clusters is also bounded by the given number $k$. So the $(k,r)$-center problem takes the $k$-center clustering objective but restricts the solution space to clusterings where every cluster has at least $r$ points. Aggarwal et al.~\cite{APFTKKZ10} give a $2$-approximation for both problems.
	
	We pose the same question, but for sum-based objectives such as $k$-median and $k$-means. 
	Here instead of the maximum distance between a point and its center, the (squared) distances are added up for all points. For a set of points $P$ from a metric space and a number $k$, the $k$-median problem is to find a clustering and centers such that the sum of the distances of every point to its closest center is minimized. For $k$-means clustering, the distances are squared,  the metric is usually Euclidean, and the centers are allowed to come from all of $\mathbb{R}^d$. 
	Now for $k$-median/$k$-means clustering with lower bounds, the situation differs in two aspects. We are given an additional parameter $B$ and solutions now satisfy the additional constraint that every cluster has at least $B$ points\footnote{In the introduction, we stick to uniform lower bounds since this is what we want for anonymity. In the technical part, we also discuss non-uniform lower bounds.}. To achieve this, points are no longer necessarily assigned to their closest center but the solution now involves an assignment function of points to centers. 
	The objective then is to minimize the (squared) sum of distances from every point to its assigned center.  
	To the best of the authors' knowledge, polynomial-time approximation algorithms for $k$-means with lower bounds have not been studied, but for $k$-median, an $O(1)$-approximation follows from known work (see below). However for both problems there are approximation algorithms running in FPT time~\cite{ BFS21, BCFN19}. 
	
	For the related (also sum-based) facility location problem, finding solutions with lower bounds on the cluster sizes appeared in very different contexts. Given sets $P$ and $F$ from a finite metric space and \emph{opening costs} for the points in $F$, the \emph{facility location problem} asks to partition $P$ into clusters and to assign a center from $F$ to each cluster such that the sum of the distances of every point to its cluster plus the sum of the opening costs of open centers is minimized. For facility location with lower bounds, an additional parameter $B$ is given and every cluster has to have at least $B$ points. 
	Karger and Minkoff~\cite{KM00} as well as Guha, Meyerson and Munagala~\cite{GMM00} use relaxed versions of facility location with (uniform) lower bounds as subroutines for solving network design problems. This inspired the seminal work of Svitkina~\cite{S10}, who gives a constant-factor approximation algorithm for the facility location problem with (uniform) lower bounds. Ahmadian and Swamy~\cite{AS12} improve the approximation ratio to 82.6. In~\cite{AS16} they state that the algorithms~\cite{S10,AS12} can be adapted for $k$-median by adequately replacing the first reduction step at the cost of an increase in the approximation factor. The adaption to $k$-median is discussed in more detail in~\cite{HHWZ20}.
	
	It is often the case that restricting the number of clusters to $k$ instead of having facility costs makes the design of approximation algorithms much more cumbersome, in particular when constraints are involved. For example, the related problem of finding a facility location solution where every cluster has to satisfy an \emph{upper} bound, usually referred to as \emph{capacitated facility location}, can be $3$-approximated (see Aggarwal et al~\cite{ALBGGGJ13}), but finding a constant-factor approximation for capacitated $k$-median clustering is a long standing open problem~\cite{ABMM019,IV20}.
	
	We demonstrate that the situation for \emph{lower} bounds is different. By a relatively straightforward approach that we borrow from the area of approximation algorithms for hierarchical clustering, we show that approximation algorithms for facility location with lower bounds can be converted into approximation algorithms for $k$-median with lower bounds (at the cost of an increase in the approximation ratio), and this reduction works also for more general $k$-clustering problems including $k$-means. This leaves us with two challenges:
	\begin{enumerate}
		\item The resulting approximation algorithm has a very high approximation ratio.
		
		\item For $k$-means clustering with lower bounds, no bicriteria or true approximation algorithm is known, and the results for standard facility location with lower bounds do not extend to the case for squared Euclidean distances: Both known algorithms for facility location use the triangle inequality an uncontrolled number of times
		to bound the cost of multiple reassignment steps. Thus the relaxed triangle inequality is not sufficient, as the resulting bound would depend on this number. Also the bicriteria algorithms by Karger and Minkoff~\cite{KM00} and Guha, Meyerson and Munagala~\cite{GMM00} require repeated application of the triangle inequality. Thus, $k$-means with lower bounds needs a new technique.
	\end{enumerate}
	
	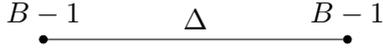
\begin{figure}
		\centering
		\begin{tikzpicture}
		\node[draw, circle, fill, inner sep = 0cm, minimum width = 1mm, label=above:{$B-1$}] (a) at (0,0) {};
		\node[draw, circle, fill, inner sep = 0cm, minimum width = 1mm, label=above:{$B-1$}] (b) at (4,0) {};
		\draw (a) to node[above] {$\Delta$} (b);
		\end{tikzpicture}
		\caption{On the difference between lower-bounded clustering and weakly lower-bounded clustering.\label{fig:smallexample}}
	\end{figure}
	
	To tackle these challenges, we define a new variation of lower-bounded clustering  that we call \emph{weakly lower-bounded} $k$-clustering.
	Here we allow points to be \emph{allocated multiple times}. However a point may not be assigned more than once to the same center. This means that our \lq clustering\rq\ is not a partitioning into subsets, but consists of not necessarily disjoint clusters (whose union is $P$). Each cluster has to respect the lower bound.
	To explain this idea, consider Figure~\ref{fig:smallexample}. 
	There are two locations with $B-1$ points each, and the distance between the two locations is $\Delta$. For clustering with a lower bound of $B$, we can only open one center, which results in a clustering cost of $(B-1)\Delta$ for $k$-median (and $\Omega(B\Delta^2)$ for $k$-means). 
	For clustering with weak lower bound $B$, we allow to assign points multiple times (but only to different centers). For each allocation, we pay the connection cost. In Figure~\ref{fig:smallexample}, this allows us to open two centers while assigning one point from every location to the other location. This costs $2\Delta$ for $k$-median (and $\Omega(\Delta^2)$ for $k$-means) for the two extra assignments. So even though we pay for more connections, the overall cost is smaller. This means that clustering with weak lower bounds can have an arbitrarily smaller cost than clustering with lower bounds, and in a way, this is a benefit: it means that we potentially pay less for having the lower bounds satisfied. Of course it also means that the gap between the optimal costs of the two problem variants with (standard) lower bounds and weak lower bounds is unbounded. We obtain the following results.
	\begin{itemize}
		\item 
		We design a $(6.5+\epsilon)$-approximation algorithm for weakly lower-bounded $k$-median and an $O(1)$-approximation algorithm for weakly lower-bounded $k$-means. The algorithms are conceptually simpler than their counterparts for lower-bounded facility location.
		\item Then we show that we can adapt the solutions such that every point is assigned to \emph{only two centers} at the cost of a constant factor increase in the approximation ratio. We say that a solution has $b$-weak lower bounds if every point is assigned to at most $b$ centers, so our results satisfy $2$-weak lower bounds. 
		\item Furthermore, we show that for $\epsilon\in (0,1)$ we can also get $O(1/\epsilon)$-approximate solutions that satisfy $(1+\epsilon)$-weak lower bounds if we allow fractional assignments of points. 
		\item Finally, we show that our result on $2$-weak lower bounds also implies a $(O(1),O(1))$-bicriteria approximation result for lower bounds, where the lower bounds are satisfied only to an extent of $B/O(1)$. Applying this result to squared Euclidean distances yields a bicriteria approximation for $k$-means with lower bounds, which is the first to the best of the authors' knowledge.
		\item Our results also extend to non-uniform lower bounds.
	\end{itemize}
	Recall our anonymization goal. When using weakly lower-bounded clustering, we still get the number of clusters that we desire and we also fully satisfy the anonymity requirement. We achieve this by distorting the data slightly by allowing data points to influence two clusters. In the fractional case, we get a solution where every data point is assigned to one main cluster and then contributes an $\epsilon$-connection to a different cluster. By this small disturbance of the data set, we can meet the anonymity lower bound requirement for all clusters. 
	
	\subparagraph*{Techniques}
	The proof that $k$-clustering can be reduced to facility location builds upon a known nesting technique from the area of approximation algorithms for hierarchical clustering and is relatively straightforward. Our conceptional contribution is the definition of weakly lower-bounded clustering as a means to achieve anonymity.
	To obtain constant-factor approximations for weakly lower-bounded clustering, the idea is to incorporate an estimate for the cost of establishing lower bounds via facility costs, approximate a $k$-clustering problem with facility costs and then enforce lower bounds on a solution by connecting the closest $B-\ell$ points to a center which previously only had $\ell$ points. Similar ideas are present in the literature, which we adapt to our new problem formulation.
	
	The main technical contribution in our paper is the proof that a solution assigning points to arbitrarily many centers can be converted into a solution where every point is assigned at most twice (or $(1+\epsilon)$-times, respectively), not only for $k$-median, but also for $k$-means. The latter means that the proof cannot use subsequent reassignment steps as it is the case in previous algorithms but has to carefully ensure that points are only reassigned once. We can also bypass this problem in the construction of a bicriteria algorithm. Previous bicriteria algorithms for lower bounds do not extend to $k$-means due to using multiple reassignments.
	
	\subparagraph*{Related work}
	Approximation algorithms for clustering have been studied for decades. The unconstrained $k$-center problem can be $2$-approximated~\cite{G85,HS86} and this is tight under P $\neq$ NP~\cite{HN79}. 
	The $(k,r)$-center problem we discussed above is introduced and $2$-approximated in~\cite{APFTKKZ10}.
	We also call this problem $k$-center with lower bounds. McCutchen and Khuller \cite{McCK08} study $k$-center with lower bounds in a streaming setting and provide a $(6+\epsilon)$-approximation. One can also consider \emph{non-uniform} lower bounds, i.e., every center has an individual lower bound that has to be satisfied if the center is opened. This variant is studied by Ahmadian and Swamy in~\cite{AS16} and they give a $3$-approximation (for the slightly more general \emph{$k$-supplier} problem with non-uniform lower bounds).
	
	The facility location problem has a rich history of approximation algorithms and the currently best algorithm due to Li~\cite{Li13}, achieving an approximation ratio of 1.488, is very close to the best known lower bound of 1.463~\cite{GK99}. Bicriteria approximation algorithms for facility location with lower bounds are developed by Karger and Minkoff~\cite{KM00} and Guha, Meyerson and Munagala~\cite{GMM00}. Svitkina~~\cite{S10} gives the first $O(1)$-approximation algorithm. The core of the algorithm is a reduction to facility location with capacities, embedded in a long chain of pre- and postprocessing steps. Ahmadian and Swamy~\cite{AS12} improve the approximation guarantee to $82.6$. For the case of non-uniform lower bounds, Li~\cite{Li19} gives an $O(1)$-approximation algorithm. Although we did not discuss this in the introduction because it is less relevant to the anonymity motivation, we show in Appendix~\ref{facloctokmedian} that this result also implies an $O(1)$-approximation for $k$-median with non-uniform lower bounds.
	
	The $k$-median and $k$-means problems are APX-hard with the best known lower bounds being $1+2/e$~\cite{JMS02} and 1.0013~\cite{ACKS15,LSW17}. The $k$-median problem can be $(2.675+\epsilon)$-approximated~\cite{BPRST17} and the best known approximation ratio for the $k$-means problem is $6.357+\epsilon$~\cite{ANSW17}. To the best of the authors' knowledge, polynomial-time approximation algorithms for $k$-means with lower bounds have not been studied before. For the $k$-median problem, $O(1)$-approximations follow relatively easily from the work on facility location as outlined in Appendix~\ref{facloctokmedian}. Furthermore there is a possible adaptation of the algorithms in~\cite{AS12, S10} as shown in ~\cite{HHWZ20}. The bicriteria algorithm for facility location with lower bounds~\cite{GMM00,KM00} can be adapted relatively straightforwardly to $k$-median with lower bounds as shown in~\cite{HHWZ2020}. The authors are neither aware of a polynomial-time approximation algorithm or bicriteria algorithm for facility location with lower bounds that works for squared metrics, nor of one for $k$-means with lower bounds. We propose a bicriteria result that is applicable to $k$-means in Appendix~\ref{sec:bicriteria}. Bera et al.~\cite{BCFN19} give a $4.676$-approximation for $k$-median with lower bounds running in time $2^k\text{poly}(n)$. For Euclidean instances in $\RR^d$, Bandyapadhyay et al.~\cite{BFS21} give a $(1+\epsilon)$-approximation for both problems running in time $2^{\tilde O(k/\epsilon^{O(1)})}\text{poly}(nd)$.
	
	Finding a polynomial constant-factor approximation algorithm for the $k$-median problem with \emph{upper bounds}, i.e., with capacities, is a long standing open problem. Recently, efforts have been made to obtain FPT approximation algorithms for the problem~\cite{ABMM019,CGKLL19}. 
	
	\section{Preliminaries} \label{sec:prelim}
	A {$k$-clustering problem} gets a finite set of input points $P$, a possibly infinite set of possible centers $F$, and a number $k \in \mathbb{N}$ and asks for a set of centers $C\subset F$ with $|C|\le k$ and a mapping $a : P \to C$ such that
	\[
	\cost(P,C,a) = \cost(C,a)=\sum_{x \in P} d(x, a(x))
	\]
	is minimized, where $d:(P \cup F) \times (P\cup F) \to \mathbb{R}^+$ is a distance function that is symmetric and satisfies that $d(x,y)=0$ iff $x=y$.
	For the \emph{generalized $k$-median problem}, the distance $d$ satisfies the $\alpha$-relaxed triangle inequality, i.e., for all $x,y,z \in P\cup F$, it holds that $d(x,y) \le \alpha d(x,z) + \alpha d(y,z)$.
	
	We define the \emph{$k$-median problem} as a generalized $k$-median problem with $P=F$ (finite) and $\alpha=1$, and the \emph{$k$-means problem} by setting $F=\mathbb{R}^d$ and $P\subset F$, and choosing $d$ as the squared Euclidean distance, for which $\alpha=2$. For these two problems, choosing the mapping $a: P \to C$ is always optimally done by assigning every point to (one of) its closest center(s). 
	A \emph{generalized facility location problem} has the same input as a generalized $k$-median problem except that it gets facility costs $f : F \to \mathbb{R}$ instead of a number $k$. The goal is to find a set of centers $C \subset F$ without cardinality constraint that minimizes
	$
	\sum_{x \in P} d(x, a(x)) + \sum_{c \in C} f(c)$.
	We use the term facility location not only if $d$ is a metric but also in the case of a distance function satisfying the $\alpha$-relaxed triangle inequality, analogously to the generalized $k$-median problem defined above.
	
	We study generalized $k$-median and generalized facility location problems under side constraints which means that the choice of the mapping $a$ is restricted.
	The side constraints that we study are versions of lower bounded clustering, i.e., they demand that every center gets a minimum number of points that are assigned to it. For clustering with (uniform) lower bounds, the input contains a number $B$ and every cluster in the solution has to have at least $B$ points. Non-uniform lower bounds are meaningful in the case of a finite set $F$ and then, non-uniform lower bounds are given via a function $B : F \to \mathbb{N}$. If any points are assigned to a center $c \in F$ in a feasible solution, then it has to be at least $B(c)$ points.
	
	When adding constraints, there is a subtle detail in the definition of generalized $k$-median problems for the case $P = F$: The question whether the center of a cluster has to be part of the cluster. Notice that without constraints, this makes no difference because assigning a center to a different center than itself cannot be beneficial. When we add lower bounds, this can change. 
	We assume that choosing a center outside of the cluster is allowed and specifically say when the solution is such that centers are members of their clusters.
	
	Our new problem variant called \emph{weakly lower-bounded generalized k-median} is defined as follows. Given an instance of the same form as for the unconstrained generalized $k$-median problem plus lower bounds $B : F \to \mathbb{N}$, the goal is to compute a set of at most $k$ centers $C\subset F$ and an assignment $a\colon P\rightarrow \mathcal P(C)$ such that the lower bound is satisfied, i.e., $|\{x\in P\mid c\in a(x)\}|\geq B(c)$ for all $c\in C$ and every point is assigned at least once. If a point is assigned multiple times the distance of the point to all assigned centers is paid by the solution. The total cost of a solution is given by  
	\[\cost(C,a)=\sum_{x\in P}\sum_{c\in a(x)}d(x, c).\]
	If a solution of a weakly lower-bounded clustering problem satisfies that every point is assigned to at most $b$ centers, then we say that the solution satisfies $b$-weak lower bounds.
	
	\section{Reducing lower-bounded \kk-clustering to facility location}
	
	In this section, we observe that by using a known technique from the area of approximation algorithms for hierarchical clustering, we can turn an approximation algorithm for generalized facility location with lower bounds into an algorithm for generalized $k$-median with lower bounds. The technique is called \emph{nesting}. Given two solutions $S_1$ and $S_2$ for the same generalized facility location problem with different number of centers $k_1 > k_2$, nesting describes how to find a solution $S$ with $k_2$ centers which has a cost bounded by a constant times the costs of $S_1$ and $S_2$ and which is \emph{hierarchically compatible} with $S_1$, i.e., the clusters in $S$ result from merging clusters in $S_1$. We use this by computing a solution $S_1$ with an approximation algorithm for generalized facility location satisfying the lower bounds and a solution $S_2$ for unconstrained generalized $k$-median and then combining them via a nesting step. The resulting solution $S$ has at most $k$ centers and the clusters result from clusters that satisfy the lower bound -- thus they satisfy the lower bound as well. For uniform lower bounds, the execution of this plan is very straightforward, for non-uniform lower bounds we have to be a bit more careful and adjust the nesting appropriately. Since most of this section follows relatively straightforwardly from known work, we defer the details to Appendix~\ref{facloctokmedian}.
	Although the reduction is applicable to generalized $k$-median, this only helps to obtain constant-factor approximations for $k$-median because no approximation algorithms for generalized facility location with lower bounds are known for $\alpha > 1$. We get the following statement from combining Lemma~\ref{lem:nesting-mergable} in Appendix~\ref{facloctokmedian} with the (adjusted) nesting results from Lin et al.~\cite{LNRW10} (see Lemma~\ref{nesting-lemma} in Appendix~\ref{facloctokmedian}) and the approximation algorithms for facility location with uniform lower bounds by Ahmadian and Swamy~\cite{AS12} and non-uniform lower bounds by Li~\cite{Li19}.
	
	\begin{corollary}
		There exist polynomial-time $O(1)$-approximation algorithms for the $k$-median problem with uniform and non-uniform lower bounds.
	\end{corollary}
	
	As a final note we observe that the crucial property of lower bound constraints we use here is \emph{mergeability}: If a uniform lower bound is satisfied for a solution, then merging clusters results in a solution that is still feasible. This is in stark contrast to for example capacitated clustering. Our reduction in Lemma~\ref{lem:nesting-mergable} works for mergeable constraints in general.
	
	\section{Generalized \kk-median with weak lower bounds} \label{sec:genkmedianwlb}
	Now we consider a relaxed version of generalized $k$-median with lower bounds where points in $P$ can be assigned multiple times. This relaxation does make sense since we have lower bounds on the centers, so it can be more valuable to assign points to multiple centers to satisfy the lower bounds instead of closing the respective centers. To see this we refer to Figure \ref{fig:smallexample}.
	We call this problem \emph{generalized k-median with weak lower bounds}.
	
	For ease of presentation, it is sensible to assume that $F$ is finite. We observe that we can always set $F=P$ at a constant increase in the cost function if we are given a uniform lower bound. In particular, we assume in this section that $F=P$ holds for $k$-means.
	
	\begin{lemma}\label{lem:f-p}
		Let $P$ be a point set and $F$ be a possibly infinite set of centers. Let $a : P \to F$ be a mapping and define $a'(x) = \arg\min_{y \in P} d(y,a(x))$. Then it holds that
		\[
		\sum_{x \in P} d(x, a'(x))  \le 2 \alpha \cdot \sum_{x \in P} d(x, a(x)).
		\]
	\end{lemma}
	\begin{proof}
		The lemma follows from the relaxed triangle inequality:
		\[
		\sum_{x \in P} d(x, a'(x)) \le \alpha \sum_{x \in P} \big (d(x, a(x))+d(a(x),a'(x))\big )
		\le 2 \alpha \cdot \sum_{x \in P} d(x, a(x)).\qedhere
		\]
	\end{proof}
	Notice that the factor can be improved for $k$-means, but here and in other places of the paper, we do not optimize the constant for $k$-means.

	To achieve anonymity it is enough to have a uniform lower bound. However if we assume $F=P$ from the beginning, then our results also hold for \emph{non-uniform lower bounds}, so we consider this more general case in this section.
	
	For standard $k$-median/$k$-means with weak lower bounds we give a $(6.5+\epsilon)$-approximate algorithm and an $O(1)$-approximate algorithm respectively. Furthermore we show that a solution to generalized $k$-median with weak lower bounds can be transformed into a solution to generalized $k$-median with \emph{2-weak lower bounds} in polynomial time. We show that this transformation increases the cost only by a factor of $\alpha(\alpha+1)$. We combine this with the approximation algorithm for standard $k$-median/$k$-means with weak lower bounds and obtain an approximation algorithm for standard $k$-median/$k$-means with 2-weak lower bounds.
	If we allow fractional assignments we show how to obtain a solution which assigns every point by an amount of at most $1+\epsilon$ for arbitrary $\epsilon\in (0,1)$, losing $\lceil\frac{1}{\epsilon}\rceil\alpha(\alpha+1)+1$ in the approximation factor. 
	
	\subparagraph*{Computing a solution}
	\label{sec:kmeanswithfaccosts}
	To approximate generalized $k$-median with weak non-uniform lower bounds, we reduce this problem to generalized $k$-median with center costs. 
	In this variant of generalized $k$-median, the input contains both a number $k$ \emph{and} center opening costs $f : F \to \mathbb{R}^+$. The objective is then 
	\[
	\cost^f(C,a) = \sum_{x \in P} d(x, a(x)) + \sum_{c \in C} f(c)
	\]
	while the solution space is constrained to center sets of size at most $k$ as for generalized $k$-median.
	
	The reduction that we use works by introducing a center cost of 
	\begin{align}
		f(c)=\sum_{p\in D_c}d(p,c)\label{eq:faccosts}
	\end{align}
	for every point $c\in F$. This cost is paid if $c$ becomes a center. Here $D_c$ is the set consisting of the $B(c)$ nearest points in $P$ to $c$. 
	The idea for this reduction is adapted from the bicriteria algorithm for lower-bounded facility location presented by Guha, Meyerson and Munagala \cite{GMM00} and Karger, Minkoff \cite{KM00}.
	
	Note that for a center $c$ in a feasible solution $(C,a)$ to generalized $k$-median with weak lower bounds, the term $\sum_{p\in D_c}d(p,c)$ is a lower bound on the assignment cost caused by $c$. This leads to the following lemma.
	\begin{restatable}{lemma}{lemGenKmedianCenterCost}\label{lem:genkmediancentercost}
		Let $OPT'$ be an optimal solution to the generalized $k$-median problem with center costs as defined in~\eqref{eq:faccosts}
		and $OPT=(O,h)$ be an optimal solution to generalized $k$-median with weak lower bounds. It holds that $\cost^f(OPT')\leq 2\cost(OPT).$ 
	\end{restatable}
	\begin{proof}
		For $p\in P$ let $c_p=\textup{argmin}\{d(p,c)\mid c\in h(p)\}$ be the closest center to which $p$ is assigned in $OPT$. We define $h'(p)=c_p$ for all $p\in P$ and obtain a feasible solution $(O, h')$ to the generalized $k$-median problem with center cost. Furthermore we have 
		\begin{align*}
		\cost^f(OPT')&\le \cost^f( O,h')=\sum_{c\in O}f(c)+\sum_{p\in P} d(p,h'(p))\\
		&=\sum_{c\in O}\sum_{p\in D_c}d(p,c)+\sum_{p\in P} d(p,h'(p))\\
		&\leq 2\sum_{p\in P}\sum_{c\in h(p)}d(p,c)\\
		&= 2\cost(OPT).
		\end{align*} 
		The second inequality follows from the fact that $\sum_{c\in O}\sum_{p\in D_c}d(p,c)$ and 
		$\sum_{p\in P}d(p,h'(p))$ are both lower bounds on the assignment cost of $OPT$.\vspace{1em}
	\end{proof}
	Let $(C,a)$ be a solution for the generalized $k$-median problem with center costs. To turn it into a solution for generalized $k$-median with weak lower bounds we have to modify the assignment. Let $c\in C$ and $n_c=|a^{-1}(c)|$. We additionally assign  $m_c=\textup{max}\{0, B(c)-n_c\}$ points to $c$ to satisfy the lower bound. Let $S_c\subset D_c$ be the set of points in $ D_c$ which are not assigned to $c$. We choose $m_c$ points from $S_c$ and assign them to $c$. This is feasible since we are allowed to assign points multiple times. Let $(C, a')$ be the corresponding solution. 
	\begin{restatable}{lemma}{lemZweitesLemma}\label{lem:zweiteslemma}
		It holds that $\cost(C,a')\leq \cost^f(C,a)$.
	\end{restatable}
	\begin{proof}
		The additional assignment cost for each center $c\in C$ can be upper bounded by $\sum_{p\in D_c}d(p,c)$. We obtain 
		\begin{align*}
		\cost(C,a')\leq& \sum_{c\in C}\sum_{p\in D_c}d(p,c)+\sum_{p\in P}d(p,a(p))\\
		=&\cost^f(C,a). \qedhere
		\end{align*}
	\end{proof}
	Lemma~\ref{lem:genkmediancentercost} and Lemma~\ref{lem:zweiteslemma} imply the following corollary.
	\begin{corollary}\label{cor:findasolution}
		Given a $\gamma$-approximation for the generalized $k$-median problem with center costs, we get a $2\gamma$-approximation for the generalized $k$-median problem with weak lower bounds in polynomial time.
	\end{corollary}
	
	We use the $(3.25+\epsilon)$-approximation for $k$-median with center costs~\cite{CL12}, which results in a $(6.5+\epsilon)$-approximation for $k$-median with weak lower bounds. 
	For $k$-means, we use the algorithm by Jain and Vazirani~\cite{Jain} which was originally designed for $k$-median. However, as outlined in the journal version~\cite{Jain}, it can be used for $k$-means when $F=P$, and also for $k$-median with center costs. The two extensions are not conflicting and can both be applied to obtain an $O(1)$-approximation for $k$-means with center costs for the case $F=P$.
	\label{sec:kmeanswithfaccosts_end}
	
	\subsection{Reducing the number of assignments per client}
	We see that the solution for standard $k$-median/$k$-means with weak lower bounds computed above can assign a point to all centers in the worst case. The number of assigned centers per point cannot be bounded by a constant. This may not be desirable in the context of publishing anonymized representatives since the distortion of the original data set is not bounded.
	
	However, we show that any solution to the generalized $k$-median problem with weak lower bounds can be transformed into a solution assigning every point at most twice. This increases the cost by a factor of $\alpha(\alpha+1)$. Recall that $\alpha$ is the constant appearing in the relaxed triangle inequality. This leads to the following theorem.
	\begin{theorem}
		\label{TwoBound}
		Given a solution $(C,a)$ to generalized k-median with weak lower bounds, we can compute a solution $(\widetilde{C}, \widetilde{a})$ to generalized k-median with 2-weak lower bounds (assigning every point at most twice) in polynomial time such that $\cost(\widetilde{C}, \widetilde{a})\leq \alpha(\alpha+1)\, \cost(C,a).$
	\end{theorem} 
	
	\subparagraph{Reassignment process.} 
	We start by setting $\widetilde C=C$ and $\widetilde a=a$ and modify both $\widetilde C$ and $\widetilde a$ until we obtain a feasible solution to generalized $k$-median with 2-weak lower bounds. During the process, the centers in $\widetilde C$ are called \emph{currently open}, and when a center is deleted from $\widetilde C$, we say it is \emph{closed}.
	The centers are processed in an arbitrary but fixed order, i.e., we assume that $C=\{c_1,\ldots,c_{k'}\}$ for some $k'\leq k$ and process them in order $c_1,\ldots,c_{k'}$. 
	We say that $c_i$ is \emph{smaller} than $c_j$ if $i < j$.
	
	Let $c=c_i$ be the currently processed center. 
	By $P_c$, we denote the set of points assigned to $c$ under $\widetilde{a}$. We divide $P_c$ into three sets $P_c^1=\{q\in P_c\mid |\widetilde{a}(q)|=1\}$, $P_c^2=\{q\in P_c\mid |\widetilde{a}(q)|=2\}$ and $P_c^3=\{q\in P_c\mid |\widetilde{a}(q)|\geq 3\}$. Furthermore with $C(P_c^3)$ we denote all centers which are connected to at least one point in $P_c^3$ under $\widetilde a$.
	 
	If $P_c^3$ is empty, we are done and proceed with the next center in $\widetilde C$. Otherwise we need to empty $P_c^3.$ Observe that points in $P_c^3$ are assigned to multiple centers, so if we delete the connection between one of these points and $c$, the point is still served by some other center. However, doing so may violate the lower bound at $c$. So we have to replace this connection.
	
	As long as $P_c^3$ is non-empty, we do the following. We pick a center $d=\min C(P_c^3)\backslash \{c\}$ and a point $x\in P_c^3$ connected to $d$. 
	We want to assign a point $y$ from $P_d^1$ to $c$ to free $x$. 
	For technical reasons, we restrict the choice of $y$: We exclude all points from the subset $\overline{P_d^1} := \{q\in P_d^1\mid |a(q)|\geq 3 \textup{ and } a(q)\cap \{c_1,\ldots, c_{i-1}\}\cap \widetilde{C}\neq \emptyset\}$, i.e., all points which were assigned to at least $3$ centers under the initial assignment $a$, and where one of these at least $3$ centers is still open \emph{and} smaller than $c$. 

	If $P_d^1\backslash \overline{P_d^1}$ is non-empty, we pick a point $y\in P_d^1\backslash \overline{P_d^1}$ arbitrarily. We set $\widetilde{a}(y)=\{d,c\}$ and $\widetilde{a}(x)=\widetilde{a}(x)\backslash\{c\}$. So $x$ is no longer connected to $c$, but to satisfy the lower bound at $c$ we replace $x$ by $y$ (Figure \ref{TwoBound:replacingX}). 
	
	\begin{figure}[b]
		\centering
		\begin{tikzpicture}[scale=0.8,dot/.style = {shape = circle, fill = black, inner sep =1.5pt},
		center/.style= {shape = rectangle, draw, minimum size= .6cm},
		every edge/.style={->, draw, >= stealth'},
		light/.style={color=gray!50!white},
		auto]
		\node[dot, label=$x$] (pointx) at (2, 1) {};
		\node[dot, label=$y$] (pointy) at (6, -1) {};
		\node[center] (centerc) at (0,0) {$c$};
		\node[center] (centerd) at (4,0) {$d$};
		
		\draw (pointx) edge (centerc)
		(pointx) edge (centerd)
		(pointy) edge (centerd);
		
		\node[dot, label=$x$] (pointx') at (2+8, 1) {};
		\node[dot, label=$y$] (pointy') at (6+8, -1) {};
		\node[center] (centerc') at (0+8,0) {$c$};
		\node[center] (centerd') at (4+8,0) {$d$};
		
		\draw (pointx') edge[light] (centerc')
		(pointx') edge (centerd')
		(pointy') edge (centerd')
		(pointy') edge (centerc');
		\end{tikzpicture}
		\caption{Connection between $x\in P_c^3$ and $c$ is deleted. A point $y\in P_d^1$ replaces $x$.}
		\label{TwoBound:replacingX}
	\end{figure}
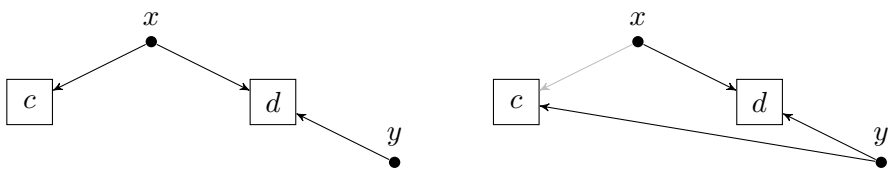
	
	If $P_d^1\backslash \overline{P_d^1}$ is empty, our replacement plan does not work. Instead, we close $d$. This means that $x$ is now assigned to one center less, and, if this happens repeatedly, $x$ will at some point no longer be in $P_c^3$. Since we close $d$, all points in $P_d^1$ have to be reassigned because they are only connected to $d$. For each $q\in P_d^1$, we reassign $q$ to the smallest currently open center in $a(q)$. Notice that such a center exists and is smaller than $c$ because $P_d^1=\overline{P_d^1}$ and for every $q \in \overline{P_d^1}$, there is at least one center in $a(q) \cap \widetilde C$ which is smaller than $c$.
	
	The entire process is described in Algorithm \ref{Alg:2WLB}. It satisfies the following invariants.
	
	\IncMargin{1ex}
	\begin{algorithm}[t]
		\caption{Reducing the number of assigned centers per point to two}
		\label{Alg:2WLB}
		\DontPrintSemicolon
		\SetKwInOut{Input}{Input}\SetKwInOut{Output}{Output}
		\SetKwFor{ForAll}{for all}{}{}
		\BlankLine
		define an ordering on the centers $c_1\leq c_2\ldots\leq c_{k'}$ \;
		set $\widetilde{C}\coloneqq C$ and $ \widetilde{a}\coloneqq a$\;
		\ForAll{$c\in C$}
		{
			$P_c\coloneqq\{q\in P\mid c\in \widetilde a(q)\}$\; 
			$P_c^3\coloneqq\{q\in P_c\mid |\widetilde{a}(q)|\geq 3\}$, \,
			$P_c^i\coloneqq\{q\in P_c\mid|\widetilde{a}(q)|=i\}$ for $i=1,2$\;
			$C(P_c^3)\coloneqq\bigcup_{q\in P_c^3}\widetilde{a}(q)$\;
			
		}
		\For{$i=1$ \KwTo $l$}
		{
			\While{$P_{c_i}^3\neq \emptyset$}
			{\label{alg:whilehead}
				$d=\textup{min } C(P_{c_i}^3)\backslash \{c_i\}$\;\label{alg:def-d}
				$\overline{P_d^1}=\{q\in P_d^1\mid |a(q)|\geq 3 \ \textup{and} \ a(q)\cap \{c_1,\ldots, c_{i-1}\}\cap \widetilde C\neq \emptyset\}\}$\;
				\If{$P_d^1\backslash\overline{P_d^1}=\emptyset$\label{alg:special}}
				{	
					\ForAll{$q\in P_d^1$} 
					{ 
						let $e=\min (a(q)\cap \widetilde{C})$\;\label{alg:reconnect1}
						set $\widetilde a(q)=\{e\}$\;\label{alg:reconnect2}
					}
					delete $d$ from $\widetilde{C}$ and all connections to $d$ in $\widetilde{a}$\;\label{alg:deletion}	
				}
				\Else{
					pick $x\in P_{c_i}^3$ connected to $d$ and $y\in P_d^1\backslash\overline{P_d^1}$\;\label{alg:def-y}
					set $\widetilde{a}(x)=\widetilde{a}(x)\backslash \{c_i\},\,\widetilde{a}(y)=\{c_i,d\}$\label{alg:swap}
					
				}
			}
			
		}
		
	\end{algorithm}
	\DecMargin{1ex}
	
	\begin{restatable}{lemma}{lemDrittesLemma}
		\label{lemma:easy_properties}
		Algorithm~\ref{Alg:2WLB} computes a feasible solution $(\widetilde C, \widetilde a)$ to generalized $k$-median with 2-weak lower bounds.	
		Furthermore the following properties hold during all steps of the algorithm. 
		\begin{enumerate}
			\item\label{lem:alg:invarianten:a0} The algorithm never establishes connections for points currently assigned more than once.
			\item\label{lem:alg:invarianten:a} 
			For any center $c \in C$, $P_c$ does not change before $c$ is  processed or closed.
			\item\label{lem:alg:invarianten:b} If a connection between $x\in P$ and the currently processed center $c\in\widetilde C$ is deleted by the algorithm, we have from this time on $x\notin P_c^3$ until termination. Moreover $P_c^3$ remains empty after $c$ is processed.
			\item\label{lem:alg:invarianten:c} While the algorithm processes $c\in C$ we always have $c<\min C(P_c^3)\backslash\{c\}$. Moreover all currently open centers which are smaller than $c$ remain open until termination.
			\item\label{lem:alg:invarianten:d} If the algorithm establishes a new connection in Line \ref{alg:reconnect2} or Line \ref{alg:swap} it remains until termination.
		\end{enumerate}
	\end{restatable}
	
	\begin{proof}
		The process terminates: For every iteration of the while loop starting in Line \ref{alg:whilehead}, either a point is deleted from $P_{c_i}^3$ or there is at least one point $x \in P_{c_i}^3$ for which $|\widetilde a(x)|$ is reduced by one. Furthermore $|\widetilde a(x)|$ does never increase for any $x\in P_{c_i}^3$.
		
		The final solution satisfies lower bounds: 
		Every time we delete a connection between a point and a center it either happens because the center is closed or we replace this connection by assigning a new point to it. So the lower bounds are satisfied at all open centers.
		
		All points stay connected to a center:
		Assume that the algorithm deletes the connection between a point $p$ and the center $d$ it is exclusively assigned to. This only happens if at this time $d$ is closed by the algorithm. Then $p$ is assigned to another center as defined in Line \ref{alg:reconnect2}. 
		
		We conclude that the solution is feasible.
		
		\textbf{Property~\ref{lem:alg:invarianten:a0}}: The algorithm establishes connections in Line \ref{alg:reconnect2} and Line \ref{alg:swap} which always involve a point currently assigned once.
		
		\textbf{Property~\ref{lem:alg:invarianten:a}:} 
		Let $c\in C$. Connections are only changed for the center that is currently processed or for a smaller center which has been processed already. Thus, the algorithm does not add or delete any connections involving $c$ before $c$ is processed or closed. 
		
		\textbf{Property~\ref{lem:alg:invarianten:b}:} 
		Assume that after the connection between $x\in P_c^3$ and $c$ is deleted by the algorithm, $x$ is again part of $P_c^3$. That would require that the algorithm establishes a new connection for a point which is connected more than once, which does not happen by Property \ref{lem:alg:invarianten:a0}. For the same reason $P_c^3$ remains empty after $c$ is processed by the algorithm. 
		
		\textbf{Property~\ref{lem:alg:invarianten:c}:}
		Assume $c$ is currently processed by the algorithm and $d=\min C(P_c^3)\backslash\{c\}$. We know that at this time $P_d^3$ is non-empty, which is by Property \ref{lem:alg:invarianten:b} only possible if $d$ is processed after $c$. Thus we have $c<d$. This also means that centers can only be closed by the algorithm if they are not processed so far.   
		
		\textbf{Property~\ref{lem:alg:invarianten:d}:}
		If a connection is deleted, the respective point is either connected to more than two centers or to a center which is closed at this time.
		A connection in Line \ref{alg:reconnect2} or Line \ref{alg:swap} is established by the algorithm between a point which is at this time assigned exactly once and a center which is already processed or currently processed by the algorithm. Thus the point is from this time on never assigned to more than two centers and the center remains open until termination by Property \ref{lem:alg:invarianten:c}. So the necessary conditions for a deletion of this connection are never fulfilled. 
	\end{proof}
	
	We now want to bound the cost of new connections created by the algorithm by the cost of the original solution. 
	Notice that only Line~\ref{alg:swap} generates new connections, Line~\ref{alg:reconnect2} re-establishes connections that were originally present.
	So let $N_c$ be the set of all points newly assigned to $c$ by the algorithm in Line \ref{alg:swap} while center $c$ is processed.
	For $y\in N_c$ let $d_y$ be the respective center in Line~\ref{alg:def-d} of Algorithm \ref{Alg:2WLB} and $x_y$ the point in Line \ref{alg:def-y} contained in $P_c^3$ and connected to $d_y$. 
	
	Using the $\alpha$-relaxed triangle inequality, we obtain the following upper bound.
	\begin{align}
		\nonumber d(y,c)	&\leq \alpha(d(y,x_y)+d(x_y,c))\leq \alpha\Big(\alpha\big (d(y,d_y)+d(d_y,x_y)\big )+d(x_y,c)\Big)\\
		&=\alpha ^2\big(d(y,d_y)+d(d_y,x_y)\big)+\alpha d(x_y,c).\label{eq:abc}
	\end{align}
	We can apply \eqref{eq:abc} to all $c\in\widetilde{C}$ and all $y\in N_c$. This yields the following upper bound on the cost of the final solution $(\widetilde C,\widetilde a)$.
	\begin{align}
		\nonumber\cost(\widetilde C,\widetilde a)		&=\sum_{c\in\widetilde{C}}\sum_{\substack{y\in P:\\ c\in\widetilde{a}(y)}} d(y,c)
		=\sum_{c\in\widetilde{C}}
		\Big (\sum_{y\in P_c\backslash N_c\hspace*{-2ex}\phantom{\widetilde C}} d(y,c)+ 
		\sum_{y\in N_c\hspace*{-1ex}\phantom{\widetilde C}} d(y,c)\Big )\\
		\label{thm:2WLB_sum}&\leq \sum_{c\in\widetilde{C}}
		\Big(\sum_{y\in P_c\backslash N_c\hspace*{-2ex}\phantom{\widetilde C}} d(y,c)+
		\sum_{y\in N_c\hspace*{-1ex}\phantom{\widetilde C}} \alpha^2(d(y,d_y)+d(d_y,x_y))+ \alpha d(x_y,c)\Big). 
	\end{align}
	
	Expression~\eqref{thm:2WLB_sum} is what we want to pay for. 
	We show in Observation~\ref{obs:all_are_in_cost} below that all involved distances contribute to the original cost as well. So in principle, we can bound each summand by a term in the original cost. But what we need to do is to bound the number of times that each term in the original cost gets charged. To organize the counting, we count how many times a specific tuple of a point $z$ and a center $f$  occurs as $d(z,f)$ in \eqref{thm:2WLB_sum}.
	Since it is important at which position a tuple appears, we give names to the different occurrences (also see Figure~\ref{fig:types}).
	
	We say that that a tuple appears as a tuple of Type 0 if it appears as $d(y,c)$ in~\eqref{thm:2WLB_sum}, as tuple of Type 1 if it appears as $d(x_y,c)$, and as tuple of Type 2 if it appears as $d(y,d_y)$ or $d(d_y,x_y)$. We distinguish the latter type further by calling a tuple occurring as $d(y,d_y)$ a tuple of Type 2.1 and a tuple occurring as $d(x_y,d_y)$ a tuple of Type 2.2. 
	We say that $(y,d_y),(d_y,x_y)$ and $(x_y,c)$ \emph{contribute} to the cost of $(y,c)$, where by the \emph{cost} of $(y,c)$ we mean the upper bound on $d(y,c)$ in \eqref{eq:abc} which we want to pay for.
	\begin{observation}
		\label{obs:all_are_in_cost}
		If a tuple $(z,f)$, $z\in P, f \in C$, occurs as Type 0, 1 or 2, then $f \in a(z) $, so in particular, $d(z,f)$ occurs as a term in the cost of the original solution.
	\end{observation}
	\begin{proof}
		For a center $c$ the set $P_c\backslash N_c$ consists of points which are assigned to $c$ by the initial assignment $a$ or assigned to $c$ while $c$ is not processed by the algorithm. The latter can only happen if a connection is reestablished in Line \ref{alg:reconnect2} which requires that the connection was already present in $(C,a)$. So Type 0 tuples satisfy the statement.
		
		For Type 1 and 2 tuples, 
		consider $y\in N_c$ for some center $c$ and the respective tuples $(x_y,c),(y,d_y),(x_y,d_y).$ Notice that both $y$ and $x_y$ are connected to $d_y$ the step before $y$ is assigned to $c$. By Property \ref{lem:alg:invarianten:c} of Lemma \ref{lemma:easy_properties} we have $c<d_y$. Thus we know by Property \ref{lem:alg:invarianten:a} of Lemma \ref{lemma:easy_properties} that $P_{d_y}$ is not changed by the algorithm at least until $y$ is assigned to $c$. So $d_y\in a(y)$ and $d_y\in a(x_y)$ which proves that Type 2 tuples satisfy the statement. 
		Moreover it holds that $c\in a(x_y)$ since there is a time where $x_y\in P_c^3$. This can, by Property \ref{lem:alg:invarianten:a0} of Lemma \ref{lemma:easy_properties}, only happen if the connection between $x_y$ and $c$ is already part of $(C,a)$. Thus, Type 1 tuples satisfy the statement.
	\end{proof}
	\begin{figure}[t]
		\centering
		\begin{tikzpicture}[scale=0.8,dot/.style = {shape = circle, fill = black, inner sep =1.5pt},
		center/.style= {shape = rectangle, draw, minimum size= .6cm},
		every edge/.style={->, draw, >= stealth'},
		light/.style={color=gray!50!white},
		auto]
		\node[dot, label=$x_y$] (pointx) at (2, 1) {};
		\node[dot, label=$y$] (pointy) at (6, -1) {};
		\node[center] (centerc) at (0,0) {$c$};
		\node[center] (centerd) at (4,0) {$d_y$};
		
		\draw (pointx) edge[light] node[swap, color = black]{$\alpha$} (centerc)
		(pointx) edge node{$\alpha^2$} (centerd)
		(pointy) edge node[swap]{$\alpha^2$} (centerd)
		(pointy) edge (centerc);
		\end{tikzpicture}
		\caption{Bounding the distance between $y$ and $c$. The respective distances appear with a factor of $\alpha$ or $\alpha^2$. Tuple $(x_y,c)$ is of Type 1 and $(x_y,d_y),(y,d_y)$ are of Type 2.\label{fig:types}}
	\end{figure}
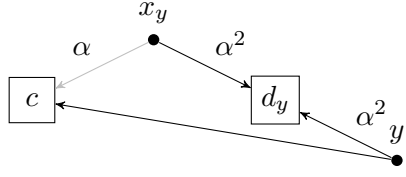
	As indicated above, a tuple $(z,f)$ can contribute to the cost of multiple tuples. Notice that a tuple occurs at most once as a tuple of Type 0 in \eqref{thm:2WLB_sum}. To bound the cost of $(\widetilde C,\widetilde a)$ we bound the number of times a tuple appears as Type 1 or Type 2 tuple in \eqref{thm:2WLB_sum}.  
	\begin{lemma}
		\label{lemma:types}
		For all $z\in P, f \in C$, the tuple $(z,f)$ can appear in \eqref{thm:2WLB_sum} at most once as a tuple of Type 1 and at most once as a tuple of Type 2.
	\end{lemma}
	\begin{proof}
		In the following, the tuple whose cost the tuple $(z,f)$ contributes to will always be named $(y,c)$, and we denote the time at which $y$ is newly assigned to $c$ by $t$.
		
		\textbf{Type 1:} Assume $(z,f)$ contributes to the cost of $(y,c)$ as a tuple of Type 1. Then $f=c$.
		Notice that at the time step before $t$ we must have $z\in P_c^3$ and afterwards, $z$ is never again contained in $P_c^3$ by Property \ref{lem:alg:invarianten:b} of Lemma \ref{lemma:easy_properties}. Thus the pair $(z,c)$ can never again be responsible for any reassignment to $c$, i.e., $(z,c)=(z,f)$ does not contribute to any further cost as a tuple of Type 1.
		
		\textbf{Type 2.1:}	
		Assume that $(z,f)$ contributes to the cost of $(y,c)$ as a tuple of Type 2.1. Then $z=y$.
		At the time step before $t$, we have $y\in P_f^1$, $f \in C(P_c^3)$, and at time $t$, we have $y\in P_c^2\cap P_{f}^2$. By Property \ref{lem:alg:invarianten:d} of Lemma \ref{lemma:easy_properties}, newly established connections are never deleted, so after time $t$, it always holds that $y\in P_{c}$.
		So even if $y$ is in $P_f$ at a later time, it cannot be in $P_f^1$ since it is also connected to $c$. So $(y,f)=(z,f)$ does not contribute to any further cost as tuple of Type 2.1.
		Furthermore by Property \ref{lem:alg:invarianten:a0} of Lemma \ref{lemma:easy_properties} we know that $y$ is always assigned to fewer than three centers after $t$ which means that $(y,f)$ does not contribute as tuple of Type 2.2 to the cost of any connection established by the algorithm after $t$ either.
		
		\textbf{Type 2.2:}	Finally we consider the case where $(z,f)$ contributes to the cost of $(y,c)$ as a tuple of Type 2.2. At time $t$, the algorithm processes $c$. By the way the algorithm chooses $f$ and $z$, we know that $z \in P_c^3$ (at the beginning of the process, i.e., before $t$) and $f = \min C(P_c^3)\backslash\{c\}$. After $t$, Property \ref{lem:alg:invarianten:b} of Lemma \ref{lemma:easy_properties} implies $z \notin P_c^3$, which means that as a tuple of Type 2.2, it can never again contribute to the cost of any tuple containing $c$.
		Assume instead that it contributes (as Type 2.2) to the cost of a tuple $(y',c')$ for a center $c' \neq c$, and some point $y' \in P$. This is supposed to happen after $t$, so $y'$ is newly assigned to $c'$ at some time $t' > t$.
		Before $c'$ is processed, we must always have $z\in P_{c'}^3$ by Property \ref{lem:alg:invarianten:a0} and \ref{lem:alg:invarianten:a} of Lemma \ref{lemma:easy_properties}. So in particular, at time $t < t'$ we have $c'\in C(P_{c}^3)\backslash \{ c\}$. Moreover we know that at some time while $c'$ is processed by the algorithm we have $f=\min C(P_{c'}^3)\backslash \{c'\}$. Using Property \ref{lem:alg:invarianten:c} of Lemma \ref{lemma:easy_properties} we conclude that $c'<f$. 
		Which is a contradiction since the algorithm chose $f$ and not $c'$ at time $t$, i.e., $f=\min C(P_{c}^3)\backslash \{ c\}$ must hold. Thus, $(z,f)$ cannot contribute to the cost of $(y',c')$ as a tuple of Type 2.2.
		
		It is left to show that $(z,f)$ cannot contribute to the cost of any $(y',c')$ as a tuple of Type 2.1 at some time $t' > t$. 
		For a contribution as Type 2.1, we would have $z=y'$ and $y' \in P_{f}^1$. We show that in this case $y'$ is in fact contained in $\overline {P_{f}^1}$. Remember that at time $t$ we have $y'=z\in P_c^3$ and that this only happens if $|a(y')|\geq 3$ by Property \ref{lem:alg:invarianten:a0} of Lemma \ref{lemma:easy_properties}. Moreover $c$ is still open by Property \ref{lem:alg:invarianten:c} of Lemma \ref{lemma:easy_properties} and is smaller than $c'$. Thus $c\in a(y')\cap \{e\mid e<c'\}\cap \widetilde C$, which proves $y'\in \overline{P_f^1}.$ Therefore the algorithm does not assign $y'$ to $c'$ (see Lines \ref{alg:special}-\ref{alg:deletion}) and $(z,f)$ does not contribute as tuple of Type 2.1 to the cost of any connection established by the algorithm after $t$. 
	\end{proof}
	
	We now know that a tuple only appears at most once as any of the three tuple types. For the final counting, we define $T0$, $T1$ and $T2$ as the sets of all tuples of Type 0, 1 and 2, respectively. 
	We could already prove a bound on the cost now, but to make it slightly smaller and prove Theorem \ref{TwoBound}, we need one final statement.
	\begin{lemma}
		\label{lemma:intersection}
		The set $T0\cap T1\cap T2$ is empty.
	\end{lemma}
	\begin{proof}
		Let $(z,f)\in T0\cap T1\cap T2$. Since $(z,f)$ is of Type 0, the point $z$ must be connected to $f$ in the final assignment $\widetilde a.$ We distinguish whether the connection between $z$ and $f$ was deleted at some point by the algorithm or not. If it is not deleted, $(z,f)$ cannot be of Type 1 since this would require that $z$ is temporarily not assigned to $f$. Otherwise the connection between $z$ and $f$ was deleted while $f$ was processed and later reestablished by the algorithm in Line \ref{alg:reconnect2}. 
		
		By assumption the tuple is also of Type 2. 
		Assume it is of Type 2.1 and contributes to the cost of a tuple $(y, c)$ with $z=y$. We know that $c<f$ by Property \ref{lem:alg:invarianten:c} of Lemma~\ref{lemma:easy_properties}. Consider the time when $z$ is newly assigned to $c$. The step before we have $z\in P_f^1$. On the other hand while $f$ is processed we have $z\in P_f^3$ in contradiction to Property \ref{lem:alg:invarianten:a0} of Lemma \ref{lemma:easy_properties}.
		
		Assume finally that $(z,f)$ is of Type 2.2 and contributes to the cost of a tuple $(y,c).$ Again we have $c<f$. Consider the time $y$ is newly assigned to $c$. The step before we have $z\in P_{c}^3$ and, by Property \ref{lem:alg:invarianten:a0} and \ref{lem:alg:invarianten:a} of Lemma \ref{lemma:easy_properties}, also $z\in P_{f}^3$. At the time the connection between $z$ and $f$ is reestablished by the algorithm, both centers are contained in $a(z)\cap\widetilde C$. This is a contradiction to $c<f=\min (a(z)\cap \widetilde C)$.   
		This completes the proof.  
	\end{proof}
	\begin{proof}[Proof of Theorem \ref{TwoBound}]
		Slightly abusing the notation we write $d(e)$ for a tuple $e=(z,f)$ by which we mean the distance $d(z,f)$. Combining Lemma \ref{lemma:types} and \ref{lemma:intersection} we obtain
		\addtocounter{equation}{-1}
		\begin{align}
		\cost(\widetilde C,\widetilde a)
		&\leq \sum_{c\in\widetilde{C}}
		\Big(\sum_{y\in P_c\backslash N_c\hspace*{-2ex}\phantom{\widetilde C}} d(y,c)+
		\sum_{y\in N_c\hspace*{-1ex}\phantom{\widetilde C}} \alpha^2(d(y,d_y)+d(d_y,x_y))+ \alpha d(x_y,c)\Big)\\
		&= \sum_{e\in T0} d(e)+\alpha^2\sum_{e\in T2}d(e)+\alpha\sum_{e\in T1}d(e) \label{thmpf:eq}\\
		&\leq (\alpha^2+\alpha)\cost (C,a).\label{thmpf:ineq} 
		\end{align}
		By Lemma \ref{lemma:types} we know that a tuple only appears at most once as any of the three tuple types. We replace \eqref{thm:2WLB_sum} by summing up the cost of all tuples in $T_i$ for $i=0,1,2$ with the respective factor for each type and obtain \eqref{thmpf:eq}. 
		
		Finally by Observation \ref{obs:all_are_in_cost} the cost $d(e)$ for $e\in T_0\cup T_1\cup T_2$ occurs as a term in the original solution and $T_0\cap T_1\cap T_2=\emptyset$ by Lemma \ref{lemma:intersection}, which proves \eqref{thmpf:ineq}.
	\end{proof}
	So it is possible to reduce the number of assignments per point to two at a constant factor increase in the approximation factor. We can go even further and allow points to be fractionally assigned to centers which poses the question if it is possible to bound the assigned amount by a number smaller than two. Indeed we can prove for every $\epsilon\in (0,1)$ that we can modify a solution to generalized $k$-median with weak lower bounds such that every point is assigned by an amount of at most $1+\epsilon$ and the cost increases by a factor of $\mathcal O(\frac{1}{\epsilon}\alpha^2)$.
	Note that even if we allow fractional assignments of points to centers, the centers remain either open or closed, which differentiates our result from a truly fractional solution, where it is also allowed to open centers fractionally. 
	Furthermore, the new assignment assigns every point to at most two centers. It is assigned by an amount of one to one center and potentially by an additional amount of $\epsilon$ to a second center. 
	
	Since we consider fractional assignments we modify our notation and denote with $\widetilde{a}_{x}^{c}\in[0,1]$ the amount by which $x\in P$ is assigned to $c\in \widetilde{C}$, where $\widetilde{C}$ is the set of centers. 
	Let $\widetilde a_x=\sum_{c\in\widetilde{C}}\widetilde a_x^c$ be the amount by which $x\in P$ is assigned to $\widetilde{C}$. 
	The assignment $\widetilde{a}$ is feasible if $\widetilde{a}_x\geq 1$ for all $x\in P$ and $\sum_{x\in P}\widetilde a_x^c\geq B(c)$ for all $c\in\widetilde{C}$, and its cost is
	\[
	\cost(\widetilde{C},\widetilde{a})
	=\sum_{c\in\widetilde{C}}\sum_{\hspace{-1ex}\phantom{\widetilde{C}}x\in P}\widetilde a_x^c d(x,c).
	\]
	The proof of the following theorem is in Appendix~\ref{appendix-eps}. It is similar to the proof of Theorem~\ref{TwoBound} but to satisfy lower bounds we can only assign an amount of $\epsilon$ from points which are already assigned once. Therefore we consider suitable sets with $\lceil\frac{1}{\epsilon}\rceil$ points, which leads to the increase of $\mathcal O(\frac{1}{\epsilon})$ in the approximation factor.  
	
	\begin{restatable}{theorem}{thmEps}
		\label{EpsilonBound}
		Given $0<\epsilon<1$ and a solution $(C,a)$ to generalized $k$-median with weak lower bounds. We can compute a solution $(\widetilde{C},\widetilde{a})$ to generalized $k$-median with $(1+\epsilon)$-weak lower bounds, i.e., $\widetilde a_x\leq 1+\epsilon$ for all $x\in P$ in polynomial time such that $\cost(\widetilde{C},\widetilde{a})\leq (\lceil\frac{1}{\epsilon}\rceil\alpha(\alpha+1)+1)\cost(C,a).$ 
	\end{restatable}
	
	On pages~\pageref{sec:kmeanswithfaccosts}-\pageref{sec:kmeanswithfaccosts_end} we reduce generalized $k$-median with weak lower bounds to generalized $k$-median with center cost and obtain a $(6.5+\epsilon)$ or $O(1)$-approximation for $k$-median or $k$-means with weak lower bounds, respectively. We combine this with Theorem \ref{TwoBound} to get a solution with $2$-weak lower bounds whose cost is a constant factor away from the problem with weak lower bounds. 
	Since weak lower bounds are a relaxation of 2-weak lower bounds, we get: 
	\begin{corollary}
		Let $OPT$ be an optimal solution to $k$-median/$k$-means with 2-weak lower bounds and $\epsilon>0$ be a constant. We can compute a solution $(C,a)$ in polynomial time for
		\begin{enumerate}
			\item k-median with 2-weak lower bounds with $\cost(C,a)\leq (13+\epsilon)\cost(OPT)$
			\item k-means with 2-weak lower bounds with $\cost(C,a)\leq O(1)\cost(OPT).$
		\end{enumerate}
	\end{corollary}
	
	Combining the results on pages~\pageref{sec:kmeanswithfaccosts}-\pageref{sec:kmeanswithfaccosts_end} with Theorem \ref{EpsilonBound} we obtain:
	\begin{corollary}
		Let $OPT$ be an optimal solution to $k$-median/$k$-means with $(1+\epsilon')$-weak lower bounds and $\epsilon>0$ be a constant. We can compute a solution $(C,a)$ in polynomial time for
		\begin{enumerate}
			\item k-median with $(1+\epsilon')$-weak lower bounds with $\cost(C,a)\leq((13+\epsilon)\lceil\frac{1}{\epsilon'}\rceil+6.5+\epsilon)\cost(OPT)$
			\item k-means with $(1+\epsilon')$-weak lower bounds with $\cost(C,a)\leq O(\frac{1}{\epsilon'})\cost(OPT).$
		\end{enumerate}
	\end{corollary}
	
	\subsection{A bicriteria algorithm to generalized \kk-median with lower bounds}
	
	A $(\beta,\delta)$-bicriteria solution for generalized $k$-median with lower bounds consists of at most $k$ centers $C'\subset F$ and an assignment $a'\colon P\rightarrow C$ such that at least $\beta B(c)$ points are assigned to $c\in C'$ by $a'$ and $\cost(C',a')\leq \delta \cost(OPT)$. Here $OPT$ denotes an optimal solution to generalized $k$-median with lower bounds.
	
	Given a $\beta \geq \frac{1}{2}$ and a $\gamma$-approximate solution to generalized $k$-median with 2-weak lower bounds $(C,a)$, we can compute a $(\beta, \gamma\max\{\frac{\alpha\beta}{1-\beta}+1, \frac{\alpha^2\beta}{1-\beta}\})$-bicriteria solution in the following way. In the beginning all points are unassigned. Let $C=\{c_1,\ldots, c_{k'}\}$. Starting at $c_1$ we decide for all centers in $C$ if they are closed or not. If we decide that a center $c$ is open we directly assign at least $\lceil\beta B(c)\rceil$ points to $c$.  Let $A_i$ denote the points assigned to $c_i$ under $a$. When considering a center $c_i$, we check if at least $\lceil\beta B(c_i)\rceil$ points in $A_i$ are not assigned so far. If so, we open $c_i$ and assign all currently unassigned points from $A_i$ to it. Otherwise, we know that a significant fraction of the points in $A_i$ are already assigned to some earlier centers. We close $c_i$ and use the existing connections to reassign the unassigned points in $A_i$. We can bound the cost of the reassignment because a constant fraction of the points in $A_i$ is already assigned. For this it is crucial that every point is assigned at most twice in the original solution. Details of the proof are in Appendix~\ref{sec:bicriteria}. 
	
	\begin{restatable}{theorem}{thmBicriteria}
		Given a $\gamma$-approximate solution $(C,a)$ to  generalized $k$-median with 2-weak lower bounds and a fixed $\beta \in [0.5,1)$, 
		Algorithm~\ref{algo:bicriteriell} (on page~\pageref{algo:bicriteriell}) computes a $(\beta, \gamma\max\{\frac{\alpha\beta}{1-\beta}+1, \frac{\alpha^2\beta}{1-\beta}\})$-bicriteria solution to generalized $k$-median with lower bounds in polynomial time. In particular, there exists a polynomial-time $(\frac{1}{2},O(1))$-bicriteria approximation algorithm for $k$-means with lower bounds.
	\end{restatable}
 
\paragraph*{Acknowledgments}
	We thank anonymous reviewers for their detailed comments to a previous version.
	Furthermore, the first author acknowledges support by DFG grant RO 5439/1-1 and the second author acknowledges support by DFG grant SCHM 2765/1-1.
	\bibliographystyle{plainurl}
	\bibliography{references}
	
	\appendix
	
	\section{Reducing lower-bounded \texorpdfstring{$k$}{k}-clustering to facility location}
\label{facloctokmedian}

Assume we want to approximate the generalized $k$-median problem under a side constraint where the side constraint is benign in the following sense: If a clustering satisfies the constraint, then the clustering resulting from merging two clusters is also feasible under the constraint. This is true for lower-bounded clustering since a cluster arising from merging two clusters with $B$ points each definitely has at least $B$ points, too. We call such constraints \emph{mergeable constraints}. A slightly weaker mergability property holds for non-uniform lower bounds where the constraint depends on the center: If we merge two clusters that satisfy lower bounds $B_1$ and $B_2$ of their centers $c_1$ and $c_2$, then the merged cluster still satisfies the lower bounds of $c_1$ and $c_2$, so as long as the merged cluster uses one of these two centers, the lower bound is still satisfied.

For many clustering problems, solving the version where the number of centers is constrained to $k$ is much more difficult to tackle than solving the facility location variant.

For example, (uniform) capacitated facility location allows for a $3$-approximation, while finding a constant-factor approximation for uniform capacitated $k$-median is a long-standing open problem. However, this is not the case for lower-bounded clustering because of the above described mergability property. 

Roughly speaking, we show that for mergable constraints, we can turn an unconstrained generalized $k$-median solution and a constrained facility location solution into a constrained generalized $k$-median solution which does not cost much more. 
To do this, we borrow a concept from the area of \emph{hierarchical} clustering which formalizes what it costs to merge clusters under a specific clustering objective. 

\begin{definition}[adapted from~\cite{LNRW10}]\label{def:nesting}
	A generalized facility location problem satisfies the $(\gamma,\delta)$-nesting property for reals $\gamma, \delta \ge 0$ if for any input point set $P$ and any two solutions $S_1=(C_1,a_1)$ and $S_2=(C_2,a_2)$ with $|C_1| > |C_2|$, a solution $S=(C,a)$ can be computed such that
	\begin{itemize}
		\item $S_1$ and $S$ are hierarchically compatible, i.e., for all $c\in C_1$ there exists a $c' \in C$ such that for all $x \in P$ with $a_1(x)=c$ it holds that $a(x)=c'$,
		\item $\cost(P,S) \le \gamma \cdot \cost(P,S_1) + \delta \cost(P,S_2)$, and
		\item $|C| \le |C_2|$.
	\end{itemize}
	We call such a solution $S$ \emph{$(\gamma,\delta)$-nested} with respect to $S_1$ and $S_2$.
\end{definition}

Lin et al.~\cite{LNRW10} show that the standard facility location / $k$-median cost function satisfies the $(2,1)$-nesting property (also see Lemma~\ref{nesting-lemma} below). Combining this with the best-known constant-factor approximation for $k$-median~\cite{BPRST17} which achieves a $2.675+\epsilon$ approximation  and the $82.6$-approximation for facility location with uniform lower bounds by Ahmadian and Swamy~\cite{AS12}, the following lemma implies a $(167.875+\epsilon)$-approximation for $k$-median with lower bounds.

\begin{lemma}\label{lem:nesting-mergable}
	Assume that we are given a generalized facility location problem that satisfies the $(\gamma,\delta)$-nesting property, an $\beta$-approximation algorithm for its generalized $k$-median variant and a $\alpha$-approximation algorithm for the constrained generalized facility location variant under a mergeable constraint. Then there is a $(\gamma \cdot \alpha+ \delta \cdot \beta)$-approximation algorithm for the generalized $k$-median problem under the same constraint.
\end{lemma}
\begin{proof}
	We compute two solutions: An $\alpha$-approximate solution $S_1=(C_1,a_1)$ for the constrained facility location variant and a $\beta$-approximate solution $S_2 = (C_2,a_2)$ for the unconstrained generalized $k$-median problem. For the facility location variant, we need no opening costs, so we set the cost of all facilities to zero.
	
	By the nesting property, we get a solution $S=(C,a)$ which costs $\cost(P,S) \le \gamma \cdot \cost(P,S_1) + \delta \cost(P,S_2)$ that is hierarchically compatible with $S_1$ and satisfies that $|C| \le |C_2| \le k$. The unconstrained generalized $k$-median problem is a relaxation of the constrained generalized $k$-median problem because all we do is drop the constraint. The constraint facility location problem without facility cost arises from dropping the condition that $|C|\le k$, so it is also a relaxation. Thus,
	\begin{align*}
		\cost(P,S) &\le \gamma \cdot \cost(P,S_1) + \delta \cost(P,S_2)\\
		&\le \gamma \cdot \alpha \cdot \cost (OPT) + \delta \cdot \beta \cdot \cost(OPT),
	\end{align*}
	where $OPT$ is an optimal solution for the constrained generalized $k$-median  problem. Since $S$ is hierarchically compatible with $S_1$, we know that every cluster in $S$ results from merging two clusters in $S_1$. Since $S_1$ satisfies the constraint and the constraint is mergeable, $S$ also satisfies the constraint.
\end{proof}

We notice one detail: Lemma~\ref{lem:nesting-mergable} implicitly assumes that we are allowed to choose centers for clusters that are not part of the cluster themselve. We stated in the introduction that we define the generalized $k$-median problem such that this is allowed. Indeed, Definition~\ref{def:nesting} above allows us to choose the set of centers $C$ such that it does not necessarily have one point from every cluster, and it allows $a$ to assign a point that is itself a cluster to a different center in $C$. For (truly) mergable constraints like uniform lower bounds, this poses no problem because the constraint is not affected by the choice of center. However, for non-uniform lower bounds, we have to be a little more careful: We need that the merged cluster is a assigned to a center whose lower bound is indeed satisfied. Definition~\ref{def:nesting} does not guarantee this. We thus prove the following slight generalization of the nesting step by Lin et al.~\cite{LNRW10} (Statement~\ref{nestinglem-itemone} only generalizes to the case of arbitrary $\alpha$, but Statement~\ref{nestinglem-itemtwo} gives the generalization that we need for non-uniform lower bounds).

\begin{lemma}\label{nesting-lemma}
	Let $S_1=(C_1,a_1)$ and $S_2=(C_2,a_2)$ be two solutions with $|C_1| > |C_2|$ for the generalized facility location problem. We can compute
	\begin{enumerate}
		\item\label{nestinglem-itemone} a solution $S=(C_2',a)$ with $C_2'\subseteq C_2$ that is $(\alpha+\alpha^2,\alpha^2)$-nested with respect to $S_1$ and $S_2$,
		\item\label{nestinglem-itemtwo} a solution $S=(C_1',a)$ with $C_1'\subseteq C_1$ that is $(\alpha^3+2\alpha^2,\alpha^3+\alpha^2)$-nested with respect to $S_1$ and $S_2$, and which satisfies that  for all $c \in C_1'$ and for all $x\in P$ with $a_1(x)=c$, it holds that $a(x)=c$.
	\end{enumerate}
\end{lemma}
\begin{proof}
	We get two solutions $S_1=(C_1,a_1)$ and $S_2=(C_2,a_2)$ with $|C_1| > |C_2|$. Let $||\cdot||$ denote the metric.
	
	For all $c_i \in C_1$, let $P_i$ be the set of all points assigned to $c_i \in C_1$ by $a_1$, and
	for all $o_j \in C_2$ , let $O_j$ be the set of all points assigned to $o_j$ by $a_2$.
	First we create a solution $S=(C_2',a)$ with $C_2' \subseteq C_2$.
	For all $i$, we assign every point $x \in P_i$ the center $o_j$ which is closest to $c_i$, i.e., $a(x) = \arg\min_{o_j \in C_2} ||c_i - o_j||$.
	By this choice we know that for any $x \in P_i$, $||c_i - o_j|| \le ||c_i - a_2(x)||$. 
	By two applications of the relaxed triangle inequality, we get that 
	\begin{align*}
		\sum_{x\in P_i} ||x-o_j|| &\le \sum_{x \in P_i} \alpha\cdot||x - c_i|| + \sum_{x \in P_i} \alpha\cdot||c_i-o_j||\\
		&\le \alpha\cdot\sum_{x \in P_i} ||x - c_i|| + \alpha\cdot\sum_{x \in P_i} ||c_i - a_2(x)||\\
		&\le \alpha\cdot\sum_{x \in P_i} ||x - c_i|| + \alpha\cdot\sum_{x \in P_i} \alpha\cdot(||c_i - x|| + ||x - a_2(x)||)\\
		&=(\alpha+\alpha^2) \cdot \sum_{x \in P_i} ||x - a_1(x)|| + \alpha^2 \cdot \sum_{x \in P_i} ||x - a_2(x)||.
	\end{align*}
	Adding the cost of all clusters yields the statement. 
	
	Now we convert $S$ into a solution $(C_1',a')$ with $C_1'\subset C_1$ at the cost of an increase in the nesting factors. Let $i$ be fixed. 
	So far, we have reassigned the points in $P_i$ to the center $o_j$ in $C_2$ closest to $c_i$. Now among all $c_{i'}$ for which $o_j$ was the closest center, we choose a center that is closest to $o_j$ and reassign the points there, i.e., $a'(x) = \arg\min \{ ||o_j-c_{i'}|| \mid a(c_{i'})=o_j\}$. The points are now assigned only to points in $C_1$. Since $a(c_{i'})=o_j$, we know that all points originally assigned to $c_{i'}$ are (re)assigned to $c_{i'}$.
	And because we only reassign a new center to the solution $C_2'$, we know that it still has at most $|C_2'|\le|C_2|$ many clusters. The cost is bounded by
	\begin{align*}
		\sum_{x\in P_i} ||x-a'(o_j))|| &\le \sum_{x\in P_i} \alpha\cdot||x-o_j|| + \alpha\cdot||o_j-a'(o_j)||\\
		&\le  \sum_{x\in P_i} \alpha\cdot||x-o_j|| + \alpha\cdot||o_j-c_i||\\
		&\le  \sum_{x\in P_i} \alpha^2 \cdot ||x - c_i|| + (\alpha^2+\alpha)\cdot||o_j-c_i||\\
		&\le  \sum_{x \in P_i} \alpha^2||x - c_i|| + (\alpha^2+\alpha)\sum_{x \in P_i} \alpha\cdot(||c_i - x|| + ||x - a_2(x)||)\\
		&=  (\alpha^3+2\alpha^2)\cdot \sum_{x \in P_i} ||x - a_1(x)|| + (\alpha^3+\alpha^2)\cdot\sum_{x \in P_i} ||x - a_2(x)||.
	\end{align*}
\end{proof}

Statement \ref{nestinglem-itemtwo} of Lemma~\ref{nesting-lemma} guarantees a solution where the centers are a subset of the centers in solution $S_1$, and the assignment ensures that points that were previously assigned to the chosen centers $C_1'$ are still assigned to their previous center. This has two benefits: a) If we previously had a solution where the centers are part of their own cluster, then this property is preserved and b) If the mergeability of the constraint depends on the center as for non-uniform lower bounds, we still satisfy the constraint. Indeed, for all $c \in C_1'$ we now know that all points previously assigned to $c$ are still assigned to $c$, then this means that if the lower bound for $c$ was satisfied by $S_1$, then it is also satisfied for $S$. Thus, we plug in the the $O(1)$-approximation for facility location with non-uniform lower bounds by Li~\cite{Li19} as $S_1$ and the already mentioned approximation $2.675+\epsilon$ approximation for $k$-median  as $S_2$ and get an $O(1)$-approximation for $k$-median with non-uniform lower bounds.

\begin{corollary}
	There exist $O(1)$-approximations for the $k$-median problem with uniform and non-uniform lower bounds.
\end{corollary}

\section{Decreasing the extra connections to an \texorpdfstring{$\epsilon$}{eps}-fraction}\label{appendix-eps}
\IncMargin{1ex}
\begin{algorithm}[t]
	\caption{Reducing the number of assigned centers per point to $1+\epsilon$\label{Alg:eWLB}}
	
	\DontPrintSemicolon
	\SetKwInOut{Input}{Input}\SetKwInOut{Output}{Output}
	\SetKwFor{ForAll}{for all}{}{}
	\BlankLine
	define an ordering on the centers $c_1< c_2\ldots< c_{k'}$ \;
	set $\widetilde{C}\coloneqq C$ and $\widetilde{a}_q^c=1$ if $c\in a(q)$ otherwise set $\widetilde a_q^c=0$\;
	\ForAll{$c\in C$}
	{	
		$P_c=\{q\in P\mid \widetilde a_q^c>0\}$\;
		$P_c^1=\{q\in P_c\mid \widetilde a_q=\widetilde a_q^c=1\}$\; 
		$P_c^{\epsilon}=\{q\in P_c\mid \widetilde a_q=1+\epsilon,\, \widetilde a_q^c=\epsilon\}$\; 
		$Q_c^\epsilon=\{q\in P_c\mid \widetilde a_q=1+\epsilon, \,\widetilde a_q^c=1\} $\; 
		$P_c^2=\{q\in P_c\mid \widetilde a_q\geq 2,\, \widetilde a_q^c=1\}$
	}
	\For{$i=1$ \KwTo $k'$}
	{
		\While{$P_{c_i}^2\neq \emptyset$}
		{	\label{alg2:start}
			$d=\textup{min } C(P_{c_i}^2)\backslash \{c_i\}$\;\label{alg2:def-d}
			$\overline {P_d^1}=P_d^1\cap \{q\in P\mid |a(q)|\geq 2 \textup{ and } a(q)\cap \{c_1,\ldots, c_{i-1}\}\cap \widetilde{C}\neq \emptyset\} $\;
			\If{$|P_d^1\backslash \overline {P_d^1}|<\frac{1}{\epsilon}$}
			{	
				delete $d$ from $\widetilde{C}$ and all connections to $d$ in $\widetilde{a}$\;	
				\ForAll{$q\in \overline {P_d^1}$} 
				{ 
					let $e=\min (a(q)\cap \widetilde{C})$\;
					set $\widetilde a_q^e=1$\;\label{alg2:reconnection}
				}
				\ForAll{$q\in Q_d^\epsilon$}
				{
					let $e\in \widetilde C$ such that $\widetilde a_q^{e}=\epsilon$\;
					set $\widetilde{a}_q^{e}=1$\;\label{alg2:reconnection2}
					
				}
				\If{$P_d^1\backslash\overline {P_d^1}\neq \emptyset$}
				{
					pick $x\in P_{c_i}^2$ connected to $d$\; \label{alg2:def-x}
					if $\widetilde a_x\geq 3$ set $\widetilde a_x^{c_i}=0$\;
					\ForAll{$q\in P_d^1\backslash \overline {P_d^1}$} 
					{
						set $\widetilde a_q^{c_i}=1$\;\label{alg2:reconnection3}
						
					}
				}
				
			}
			\Else{
				pick $x\in P_{c_i}^2$ connected to $d$ and $A\subset P_d^1\backslash \overline {P_d^1}$ of cardinality $\lceil\frac{1}{\epsilon}\rceil$\;\label{alg2:def-x2}
				set $\widetilde{a}_x^{c_i}=0$ and $\widetilde{a}_y^{c_i}=\epsilon$ for all $y\in A$\; \label{alg2:swap}
				
			}
		}
		
	}
\end{algorithm}
\DecMargin{1ex}

\thmEps*

\subparagraph{Reassignment process.}
In the beginning we set $\widetilde C=C$. For $q\in P$ let $\widetilde a_q^c=1$ if $c\in a(q)$ and otherwise let $\widetilde a_q^c=0$. We modify both $\widetilde C$ and $\widetilde a$ until we obtain a valid solution to generalized $k$-median with $(1+\epsilon)$-weak lower bounds. During the process, the centers in $\widetilde C$ are called \emph{currently open}, and when a center is deleted from $\widetilde C$, we say it is \emph{closed}. The centers are processed in an arbitrary but fixed order, i.e., we assume that $C=\{c_1,\ldots,c_{k'}\}$ for some $k'\leq k$ and process them in order $c_1,\ldots,c_{k'}$. 
We say that $c_i$ is \emph{smaller} than $c_j$ if $i < j$.

Before we start explaining the reassignment we observe that the following properties hold for $(\widetilde C,\widetilde a)$ in the beginning. 
\begin{enumerate}
	\item for all $q\in P$ we have either  $\widetilde a_q\in\mathbb N$ or $\widetilde a_q=1+\epsilon.$
	\item if $\widetilde a_q=1+\epsilon$ then $q$ is assigned to one center by an amount of one and to a second center by an amount of $\epsilon.$
	\item if $\widetilde a_q\in \mathbb N$ then $\widetilde a_q^c\in\{0,1\}$ for all $c\in C$.
\end{enumerate}
We ensure that these properties also hold during the whole reassignment process.

Let $c=c_i$ be the currently processed center.
By $P_c$ we denote the set of points assigned to $c$ by a positive amount under $\widetilde a.$ We divide $P_c$ into the four sets $P_c^1=\{q\in P_c\mid \widetilde a_q=\widetilde a_q^c=1\}$, $P_c^{\epsilon}=\{q\in P_c\mid \widetilde a_q=1+\epsilon, \widetilde a_q^c=\epsilon\}$, $Q_c^{\epsilon}=\{q\in P_c\mid \widetilde a_q=1+\epsilon, \widetilde a_q^c=1\}$ and finally $P_c^2=\{q\in P_c\mid \widetilde a_q\geq 2,\widetilde a_q^c=1\}$. Thus we differentiate between points which are assigned exclusively to $c$, points which are assigned by an amount of $\epsilon$ to $c$ and by an amount of one to an other center or vice versa and points which are assigned by an amount of one to $c$ and by an amount of at least one to some other centers. Furthermore with $C(P_c^2)$ we denote all centers which are connected to at least one point in $P_c^2$ under $\widetilde a$. Observe that indeed $P_c=P_c^1\cup P_c^{\epsilon}\cup Q_c^{\epsilon}\cup P_c^2$ if the above properties hold at that time.      

Notice that points in $P_c\backslash P_c^2$ are already assigned by an amount of at most $1+\epsilon$, so we only care about points in $P_c^2$. If $P_c^2$ is empty, we are done and proceed with the next center in $\widetilde C$. Otherwise we need to empty $P_c^2.$ Observe that points in $P_c^2$ are assigned to multiple centers, so if we delete the connection between one of these points and $c$, the point is still served by some other center. However, doing so violates the lower bound at $c$. So we have to replace this connection. 

As long as $P_c^2$ is non-empty, we do the following. We pick a center $d=\min C(P_c^2)\backslash \{c\}$ and a point $x\in P_c^2$ connected to $d$. 
We want to assign points from $P_d^1$ by amount of $\epsilon$ to $c$ to free $x$.
For technical reasons, we restrict the choice of these points: We exclude all points from the subset $\overline{P_d^1} := \{q\in P_d^1\mid |a(q)|\geq 2 \textup{ and } a(q)\cap \{c_1,\ldots, c_{i-1}\}\cap \widetilde{C}\neq \emptyset\}$, i.e., all points which were assigned to at least $2$ centers under the initial assignment $a$, and where one of these at least $2$ centers is still open \emph{and} smaller than $c$.

We can only assign points from $P_d^1\backslash \overline{P_d^1}$ to $c$ if its cardinality is at least $\lceil \frac{1}{\epsilon}\rceil$. If this is the case we choose a set $A$ of $\lceil\frac{1}{\epsilon}\rceil$ points from $P_d^1\backslash \overline{P_d^1}$ and set $\widetilde a_q^c=\epsilon$ for all $q\in A$ . Furthermore we set $\widetilde a_x^c=0$. So $x$ is no longer connected to $c$, but to satisfy the lower bound at $c$ we replace $x$ by a set of $\lceil\frac{1}{\epsilon}\rceil$ points which are now connected to $c$ by an amount of $\epsilon$ (Figure \ref{eBound:replacingX}). By this we guarantee that the lower bound at $c$ is still satisfied.

If $|P_d^1\backslash \overline{P_d^1}|<\lceil\frac{1}{\epsilon}\rceil$ our replacement plan does not work. Instead we close $d$ and set $\widetilde a_q^d=0$ for all $q\in P.$ If we close $d$, points in $P_d^1\cup Q_d^{\epsilon}$ will be assigned by an amount smaller than one, thus we do the following. 
All points in $P_d^1\backslash \overline{P_d^1}$ are reassigned to $c$, i.e., $\widetilde a_q^c=1$ for $q\in P_d^1\backslash\overline{P_d^1}$ (Figure \ref{eBound:DeletionD}). Since we assign all points in $P_d^1\backslash \overline{P_d^1}$ to $c$, we could delete this many connections between clients in $P_c^2$ and $c$. But for simplicity, if $P_d^1\backslash \overline{P_d^1}$ is non-empty and $\widetilde a_x\geq 3$, we only delete the connection between $x$ and $c$. A point $q\in \overline{P_d^1}$ is reassigned to the smallest open center in $a(q)$ by an amount of one. And finally every point in $Q_d^\epsilon$ is assigned by an amount of $\epsilon$ to some other center than $d$, so we add an additional amount of $1-\epsilon$ to this assignment. 

Observe that none of the above reassignments violates the claimed properties for $(\widetilde C,\widetilde a)$ above.
The entire procedure is described in Algorithm \ref{Alg:eWLB}.

\begin{lemma}
	\label{lemma:easy_properties2}
	Algorithm~\ref{Alg:eWLB} computes a feasible solution $(\widetilde C, \widetilde a)$ to generalized $k$-median with $(1+\epsilon)$-weak lower bounds.	
	Furthermore the following properties hold during all steps of the algorithm. 
	\begin{enumerate}
		\item\label{lem:alg:invarianten:a2} 
		For any center $c \in C$, $P_c$ does not change before $c$ is processed or closed. Up to that point all points in $P_c$ are assigned by an amount of $1$ to $c$.
		\item\label{lem:alg:invarianten:b2} If a connection between $x\in P$ and the currently processed center $c\in\widetilde C$ is deleted by the algorithm, we have from this time on $x\notin P_c^2$ until termination. Moreover $P_c^2$ remains empty after $c$ is processed.
		\item\label{lem:alg:invarianten:c2} While the algorithm processes $c\in C$ we always have $c<\min C(P_c^2)\backslash\{c\}$. Moreover all currently open centers which are smaller than $c$ remain open until termination.
		\item\label{lem:alg:invarianten:d2} If the algorithm establishes a new connection in Line~\ref{alg2:reconnection}, Line~\ref{alg2:reconnection3} or Line~\ref{alg2:swap} it remains until termination.
	\end{enumerate}
\end{lemma}
\begin{proof}
	The process terminates: For every iteration of the while loop starting in Line~\ref{alg2:start}, either a point is deleted from $P_{c_i}^2$ or there is at least one point $x \in P_{c_i}^2$ for which $\widetilde a_x$ is reduced by one. Furthermore $\widetilde a_x$ does never increase for any $x\in P_{c_i}^2$.
	
	The final solution satisfies lower bounds: 
	Every time we delete a connection between a point and a center it either happens because the center is closed or we replace this connection by assigning $\lceil\frac{1}{\epsilon}\rceil$ new points each by an amount of $\epsilon$ to it. So the lower bounds are satisfied at all open centers.
	
	All points are assigned by an amount of at least 1:
	Assume that the algorithm deletes the connection between a point $p$ and a center $d$. This either happens if $p$ is assigned by a total amount of at least $2$ at this time or $d$ is closed by the algorithm. In the last case we ensure in Line~\ref{alg2:reconnection}, Line~\ref{alg2:reconnection2} or Line~\ref{alg2:reconnection3} that $p$ is assigned by an amount of one to an other center after we close $d$. 
	
	All points are assigned by an amount of at most $1+\epsilon$: For $c\in C$ we know by Property \ref{lem:alg:invarianten:b2} that $P_c^2$ is empty after termination. Then $P_c=P_c^1\cup P_c^{\epsilon} \cup Q_c^{\epsilon}$, so all points connected to $c$ are assigned by a total amount of at most $1+\epsilon.$ 
	
	We conclude that the solution is feasible.
	
	\textbf{Property~\ref{lem:alg:invarianten:a2}:} 
	Let $c\in C$. Assume the property is true up to a time $t$. In the next step connections may change for the center that is currently processed, for a smaller center which has been processed already or for a center which is currently connected to a point by an amount of $\epsilon$. If $c$ is not processed so far none of this applies to it, so the property also holds in the next step.  
	
	\textbf{Property~\ref{lem:alg:invarianten:b2}:} 
	Assume that after the connection between $x\in P_c$ and $c$ is deleted by the algorithm, $x$ is part of $P_c^2$. That would require that the algorithm assigns $x$ to a center by an amount of one while it is already assigned to a second center by an amount of one, which does not happen. For the same reason $P_c^2$ remains empty after $c$ is processed by the algorithm. 
	
	\textbf{Property~\ref{lem:alg:invarianten:c2}:}
	Assume $c$ is currently processed by the algorithm and $d=\min C(P_c^2)\backslash\{c\}$. We know that at this time $P_d^2$ is non-empty. Which is by Property \ref{lem:alg:invarianten:b2} only possible if $d$ is processed after $c$. Thus we have $c<d$. This also means that centers can only be closed by the algorithm if they are not processed so far.   
	
	\textbf{Property~\ref{lem:alg:invarianten:d2}:}
	A connection established in Line~\ref{alg2:reconnection} involves a center which is already processed by the algorithm. By Property \ref{lem:alg:invarianten:c2} such centers remain open, thus the connection is not deleted until termination. In Line~\ref{alg2:reconnection3} and Line~\ref{alg2:swap} the algorithm establishes a connection between the currently processed center $c$ and some point $p$ which is assigned by an amount of at most 1 at this time. If this connection is deleted at some later point in time, this would require that $c$ is closed by the algorithm or $p\in P_c^2$. Both can not happen.   
\end{proof}

We bound the cost of $(\widetilde C,\widetilde a)$ in a similar way we bounded the cost of the solution in Theorem~\ref{TwoBound}. Let $N_c$ denote the set of points which are newly assigned by $\epsilon$ respectively $1$ to $c$ while $c$ is processed. This happens in Line~\ref{alg2:reconnection3} and Line~\ref{alg2:swap} of the algorithm. We want to charge the cost of these new connections to the cost of the original solution. 

For $y\in N_c$ let $d_y$ be the respective center in Line~\ref{alg2:def-d} of Algorithm \ref{Alg:eWLB} and $x_y$ the point in Line~\ref{alg2:def-x} respectively Line~\ref{alg2:def-x2} contained in $P_c^2$ and connected to $d_y$. 
Using the $\alpha$-relaxed triangle inequality, we obtain the following upper bound.
\begin{align}
	\nonumber d(y,c)	&\leq \alpha(d(y,x_y)+d(x_y,c))\leq \alpha\Big(\alpha\big (d(y,d_y)+d(d_y,x_y)\big )+d(x_y,c)\Big)\\
	&\leq \alpha ^2\big(d(y,d_y)+d(d_y,x_y)\big)+\alpha d(x_y,c).\label{eq:abc2}
\end{align}
We can apply \eqref{eq:abc2} to all $c\in\widetilde{C}$ and all $y\in N_c$. This yields the following upper bound on the cost of the final solution $(\widetilde C,\widetilde a)$.
\begin{align}
	\nonumber\cost(\widetilde C,\widetilde a)
	&=\sum_{c\in \widetilde C}\sum_{\hspace*{-1ex}\phantom{\widetilde C}y\in P}d(y,c)\widetilde{a}_y^c
	\leq \sum_{c\in \widetilde C}\Big (\hspace*{-1ex} \sum_{\hspace*{-1ex}\phantom{\widetilde C}y\in P_c\backslash N_c}d(y,c)+\sum_{\hspace*{-1ex}\phantom{\widetilde C}y\in N_c}d(y,c)\Big)\\
	\label{thm:eWLB_sum} &\leq \sum_{\hspace*{-1ex}\phantom{\widetilde C} c\in \widetilde C}\Big(\hspace*{-1ex}\sum_{\hspace*{-1ex}\phantom{\widetilde C}y\in P_c\backslash N_c} d(y,c)+
	\sum_{y\in N_c\hspace*{-1ex}\phantom{\widetilde C}} \alpha^2(d(y,d_y)+d(d_y,x_y))+ \alpha d(x_y,c)\Big). 
\end{align}	  

Notice that in the first inequality we use the fact that $\widetilde a_y^c\leq 1$. So we pay the the price of connecting $y$ to $c$ by an amount of $1$ independent of whether $\widetilde a_y^c$ is $1$ or $\epsilon.$ 

Expression~\eqref{thm:eWLB_sum} is what we want to pay for. 
Observe that all involved distances contribute to the original cost as well (we state this formally in Observation \ref{obs:all_are_in_cost2} below). So in principle, we can charge each summand to a term in the original cost. But what we need to do is to bound the number of times that each term in the original cost gets charged. To organize the counting, we count how many times a specific tuple of a point $z$ and a center $f$  occurs as $d(z,f)$ in \eqref{thm:eWLB_sum}.
Since it is important at which position a tuple appears, we give names to the different occurrences.
We say that that a tuple appears as a tuple of Type 0 if it appears as $d(y,c)$ in~\eqref{thm:eWLB_sum}, as tuple of Type 1 if it appears as $d(x_y,c)$, and as tuple of Type 2 if it appears as $d(y,d_y)$ or $d(d_y,x_y)$. We distinct the latter type further by calling a tuple occurring as $d(y,d_y)$ a tuple of Type 2.1 and a tuple occurring as $d(x_y,d_y)$ a tuple of Type 2.2. 
We say that $(y,d_y),(d_y,x_y)$ and $(x_y,c)$ \emph{contribute} to the cost of $(y,c)$, where by the \emph{cost} of $(y,c)$ we mean the upper bound on $d(y,c)$ in ~\eqref{eq:abc2} which we want to pay for.

\begin{observation}
	\label{obs:all_are_in_cost2}
	If a tuple $(z,f)$, $z\in P, f \in C$, occurs as Type 0, 1 or 2, then $f \in a(z) $, so in particular, $d(z,f)$ occurs as a term in the cost of the original solution.
\end{observation}
\begin{proof}
	For a center $c$ the set $P_c\backslash N_c$ contains points which are assigned to $c$ by the initial assignment $a$ or assigned to $c$ while $c$ is not processed by the algorithm. Latter can only happen if a connection is reestablished in Line~\ref{alg2:reconnection} which requires that the connection was already present in $(C,a)$. So Type 0 tuples satisfy the statement.
	
	For Type 1 and 2 tuples, 
	consider $y\in N_c$ for some center $c$ and the respective tuples $(x_y,c),(y,d_y)$, $(x_y,d_y).$ Notice that both $y$ and $x_y$ are connected to $d_y$ before $y$ is assigned to $c$. By Property \ref{lem:alg:invarianten:c2} of Lemma \ref{lemma:easy_properties2} we have $c<d_y$. Thus we know by Property \ref{lem:alg:invarianten:a2} of Lemma \ref{lemma:easy_properties2} that $d_y\in a(y)$ and $d_y\in a(x_y)$ which proves that Type 2 tuples satisfy the statement. 
	Moreover it holds that $c\in a(x_y)$ since there is a time where $x_y\in P_c^2$, which can only happen if the connection between $x_y$ and $c$ is already part of $(C,a)$. Thus, Type 1 tuples satisfy the statement.
\end{proof}

As indicated above, a tuple $(z,f)$ can contribute to the cost of multiple tuples. Notice that a tuple occurs at most once as a tuple of Type 0 in \eqref{thm:eWLB_sum}. To bound the cost of $(\widetilde C,\widetilde a)$ we bound the number of times a tuple appears as Type 1 or Type 2 tuple in \eqref{thm:eWLB_sum}. 

Remember that we used a similar statement in the proof of Theorem \ref{TwoBound}, where we proved that every tuple can appear at most once as each type. However here we can only bound the appearance by $\lceil\frac{1}{\epsilon}\rceil$ for Type 1 and Type 2 tuples due to Line~\ref{alg2:reconnection3} and Line~\ref{alg2:swap} where we assign up to $\lceil\frac{1}{\epsilon}\rceil$ points from $P_d^1$ to $c$. Notice that even if we assign each of these points initially by an amount of $\epsilon$ to $c$ as it is done in Line~\ref{alg2:swap}, that amount can be increased to $1$ at some later time in Line~\ref{alg2:reconnection2}. The proof is similar to that of Lemma \ref{lemma:types} but we carry out the arguments again for sake of completeness.
\begin{lemma}
	\label{lemma:types2}
	For all $z\in P, f \in C$, the tuple $(z,f)$ appears in \eqref{thm:eWLB_sum} at most $\lceil\frac{1}{\epsilon}\rceil$ times as tuple of Type 1 and at most $\lceil\frac{1}{\epsilon} \rceil$ times as tuple of Type 2.
\end{lemma}
\begin{proof}
	In the following, the tuple whose cost the tuple $(z,f)$ contributes to will always be named $(y,c)$, and we denote the time at which $y$ is newly assigned to $c$ by $t$.
	
	\textbf{Type 1:} Assume $(z,f)$ contributes to the cost of $(y,c)$ as a Tuple of Type 1. Then $f=c$. 
	At $t$ we assign up to $\lceil\frac{1}{\epsilon}\rceil$ points to $c$. So $(z,f)$ contributes to the cost of at most $\lceil\frac{1}{\epsilon}\rceil$ connections established by the algorithm at $t$ as tuple of Type 1.
	Notice that at the time step before $t$ we must have $z\in P_c^2$ and afterwards, $z$ is never again contained in $P_c^2$ by Property \ref{lem:alg:invarianten:b2} of Lemma \ref{lemma:easy_properties2}.  Thus the tuple $(z,c)$ can not be responsible for any assignment to $c$ after $t$, i.e., $(z,c)=(z,f)$ does not contribute to any further cost as a tuple of Type 1.
	
	\textbf{Type 2.1:}	
	Assume that $(z,f)$ contributes to the cost of $(y,c)$ as a Tuple of Type 2.1. Then $z=y$.
	At the time step before $t$, we have $y\in P_f^1$, $f \in C(P_c^2)$. By Property \ref{lem:alg:invarianten:d2} of Lemma \ref{lemma:easy_properties2}, newly established connections stay, so after time $t$, it always holds that $y\in P_{c}$.
	So even if $y$ is in $P_f$ at a later time, it can not be in $P_f^1$ since it is also connected to $c$. So $(y,f)=(z,f)$ does not contribute to any further cost as tuple of Type 2.1.
	Furthermore, observe that the algorithm never adds a connection to a point which is assigned more than once. So we know that $y$ is always assigned by an amount of at most $1+\epsilon$ after $t$ which means that $(y,f)$ does not contribute as tuple of Type 2.2 to the cost of any connection established by the algorithm after $t$ either.
	
	\textbf{Type 2.2:}	Finally we consider the case where $(z,f)$ contributes to the cost of $(y,c)$ as a tuple of Type 2.2. At time $t$, the algorithm processes $c$. By the way the algorithm chooses $f$ and $z$, we know that $z \in P_c^2$ (at the beginning of the process, i.e., before $t$) and $f = \min C(P_c^2)\backslash\{c\}$. After $t$, Property \ref{lem:alg:invarianten:b2} of Lemma \ref{lemma:easy_properties2} implies $z \notin P_c^2$, which means that as a tuple of Type 2.2, it can not contribute to the cost of any tuple containing $c$ after $t$. However it contributes as tuple of Type 2.2 to the cost of up to $\lceil\frac{1}{\epsilon}\rceil-1$ additional connections at time $t$ (see Line~\ref{alg2:reconnection3} and Line \ref{alg2:swap}).
	Assume instead that it contributes (as Type 2.2) to the cost of a tuple $(y',c')$ for a center $c' \neq c$, and some point $y' \in P$. This is supposed to happen after $t$, so $y'$ is newly assigned to $c'$ at some time $t' > t$.
	The step before $t'$ we have $z\in P_{c'}^2$. Thus before $c'$ is processed, we must always have $z\in P_{c'}^2$ by Property \ref{lem:alg:invarianten:a2} of Lemma \ref{lemma:easy_properties2}. So in particular, at time $t < t'$ we have $c'\in C(P_{c}^2)\backslash \{ c\}$. Moreover we know that at some time while $c'$ is processed by the algorithm we have $f=\min C(P_{c'}^2)\backslash \{c'\}$. Using Property \ref{lem:alg:invarianten:c2} of Lemma~\ref{lemma:easy_properties2} we conclude that $c'<f$. 
	Which is a contradiction since the algorithm chose $f$ and not $c'$ at time $t$, i.e., $f=\min C(P_{c}^2)\backslash \{ c\}$ must hold. Thus, $(z,f)$ can not contribute to the cost of $(y',c')$ as a tuple of Type 2.2.
	
	It is left to show that $(z,f)$ can not contribute to the cost of any $(y',c')$ as a tuple of Type 2.1 at some time $t' > t$. 
	For a contribution as Type 2.1, we would have $z=y'$ and $y' \in P_{f}^1$. We show that in this case $y'$ is even contained in $\overline {P_{f}^1}$. Remember that at time $t$ we have $y'=z\in P_c^2$ and that this only happens if $|a(y')|\geq 2$. Moreover $c$ is sill open by Property \ref{lem:alg:invarianten:c2} of Lemma \ref{lemma:easy_properties2} and is smaller than $c'$. Thus $c\in a(y')\cap \{e\mid e<c'\}\cap \widetilde C$, which proves $y'\in \overline{P_f^1}.$ Therefore the algorithm does not assign $y'$ to $c'$ (see Line \ref{alg2:reconnection}) and $(z,f)$ does not contribute as tuple of Type 2.1 to the cost of any connection established by the algorithm after $t$. 
\end{proof}
For the final counting, we define $T0$, $T1$ and $T2$ as the sets of all tuples of Type 0, 1 and 2, respectively. 

\begin{proof}[Proof of Theorem \ref{TwoBound}]
	Slightly abusing the notation we write $d(e)$ for a tuple $e=(z,f)$ by which we mean the distance $d(z,f)$. We obtain
	\addtocounter{equation}{-1}
	\begin{align}
		\cost(\widetilde C,\widetilde a)
		&\leq \sum_{c\in\widetilde{C}}
		\Big(\sum_{y\in P_c\backslash N_c\hspace*{-2ex}\phantom{\widetilde C}} d(y,c)+
		\sum_{y\in N_c\hspace*{-1ex}\phantom{\widetilde C}} \alpha^2(d(y,d_y)+d(d_y,x_y))+ \alpha d(x_y,c)\Big)\\
		&= \sum_{e\in T0} d(e)+\alpha^2\Big\lceil\frac{1}{\epsilon}\Big\rceil\sum_{e\in T2}d(e)+\alpha\Big\lceil\frac{1}{\epsilon}\Big\rceil\sum_{e\in T1}d(e) \label{thmpf:eq2}\\
		&\leq \big (\Big\lceil\frac{1}{\epsilon}\Big\rceil\alpha(\alpha+1)+1\big)\cost (C,a).\label{thmpf:ineq2} 
	\end{align} 
	Here we replace \eqref{thm:eWLB_sum} by summing up the cost of all tuples in $T_i$ for $i=0,1,2$ with the respective factor times the maximal number of appearances for each type. Thus by Lemma \ref{lemma:types2} we obtain a total factor of 1 for Type 0, $\alpha^2 \lceil\frac{1}{\epsilon}\rceil$ for Type 1 and $\alpha \lceil\frac{1}{\epsilon}\rceil$ for Type 2 (see \eqref{thmpf:eq2}). 
	
	Finally by Observation \ref{obs:all_are_in_cost2} the cost $d(e)$ for $e\in T_0\cup T_1\cup T_2$ occurs as a term in the original solution which proves \eqref{thmpf:ineq2}.
\end{proof}

\begin{figure}[t]
	\centering
	\resizebox{1\linewidth}{!}{
		\begin{tikzpicture}[dot/.style = {shape = circle, fill = black, inner sep =1.5pt},
		center/.style= {shape = rectangle, draw, minimum size= .6cm},
		every edge/.style={->, draw, >= stealth'},
		light/.style={color=gray!50!white},
		auto]
		\node[dot, label=$x$] (pointx) at (2, 1) {};
		\node[dot] (point1) at (5, -2.7) {};
		\node[dot] (point2) at (5.5, -2.35) {};
		\node[dot] (point3) at (6, -2) {};
		\draw[rotate=35] (5.5*0.819-2.35*0.573, -5.5*0.573-2.35*0.819) ellipse (1 and 0.5);
		\node at (7.2, -2.35) {$P_d^1\setminus\overline {P_d^1}$};
		\node[center] (centerc) at (0,0) {$c$};
		\node[center] (centerd) at (4,0) {$d$};
		
		\draw (pointx) edge node[swap]{1} (centerc)
		(pointx) edge node{1} (centerd)
		(point1) edge node{1} (centerd)
		(point2) edge (centerd)
		(point3) edge node[swap]{1} (centerd);
		\end{tikzpicture}
		\begin{tikzpicture}[dot/.style = {shape = circle, fill = black, inner sep =1.5pt},
		center/.style= {shape = rectangle, draw, minimum size= .6cm},
		every edge/.style={->, draw, >= stealth'},
		every node/.style={color=black},
		light/.style={color=gray!50!white},
		auto]
		\node[dot, label=$x$] (pointx) at (2, 1) {};
		\node[dot] (point1) at (5, -2.7) {};
		\node[dot] (point2) at (5.5, -2.35) {};
		\node[dot] (point3) at (6, -2) {};
		\draw[rotate=35] (5.5*0.819-2.35*0.573, -5.5*0.573-2.35*0.819) ellipse (1 and 0.5);
		\node at (6.8, -2.35) {$A$};
		\node[center] (centerc) at (0,0) {$c$};
		\node[center] (centerd) at (4,0) {$d$};
		
		\draw (pointx) edge[light] (centerc)
		(pointx) edge node{1} (centerd)
		(point1) edge (centerd)
		(point2) edge (centerd)
		(point3) edge node[swap]{1} (centerd)
		(point1) edge node{$\epsilon$} (centerc)
		(point2) edge (centerc)
		(point3) edge node[swap]{$\epsilon$} (centerc);
		\end{tikzpicture}}
	\caption{Shows case $|P_d^1\backslash \overline {P_d^1}|>\lceil\frac{1}{\epsilon}\rceil$. Pick a set $A\subset P_d^1\backslash\overline {P_d^1}$ of cardinality $\lceil\frac{1}{\epsilon}\rceil$ and assign an amount of $\epsilon$ from points in $A$ to $c$. Here $A=P_d^1\backslash \overline {P_d^1}.$}
	\label{eBound:replacingX}
\end{figure}
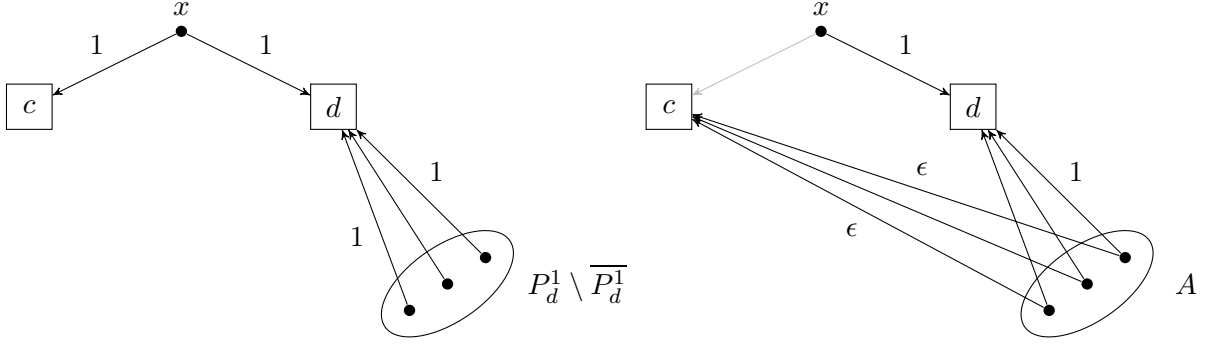 
\begin{figure}[t]
	\centering
	\begin{tikzpicture}[dot/.style = {shape = circle, fill = black, inner sep =1.5pt},
	center/.style= {shape = rectangle, draw, minimum size= .6cm},
	every edge/.style={->, draw, >= stealth'},
	light/.style={color=gray!50!white},
	auto]		
	\node[dot, label=$x$] (pointx') at (2+9, 1) {};
	\node[dot] (point1') at (5+9, -2.7) {};
	\node[dot] (point2') at (5.5+9, -2.35) {};
	\node[dot] (point3') at (6+9, -2) {};
	\draw[rotate=35] (5.5*0.819+9*0.819-2.35*0.573, -5.5*0.573-9*0.573-2.35*0.819) ellipse (1 and 0.5);
	\node at (7.2+9, -2.35) {$P_d^1\setminus \overline {P_d^1}$};
	\node[center] (centerc') at (0+9,0) {$c$};
	\node[center] (centerd') at (4+9,0) {$d$};
	
	\draw (pointx') edge[light] (centerc')
	(pointx') edge[light] (centerd')
	(point1') edge[light] (centerd')
	(point2') edge[light] (centerd')
	(point3') edge[light] (centerd')
	(point1') edge node{1} (centerc')
	(point2') edge (centerc')
	(point3') edge (centerc');
	
	\end{tikzpicture}
	\caption{Shows case $|P_d^1\backslash \overline {P_d^1}|<\lceil \frac{1}{\epsilon}\rceil$. Center $d$ is closed and points from $P_d^1\backslash \overline {P_d^1}$ are assigned to $c$.}
	\label{eBound:DeletionD}
\end{figure}
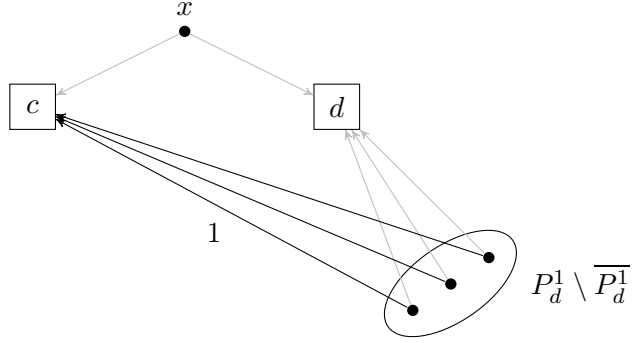
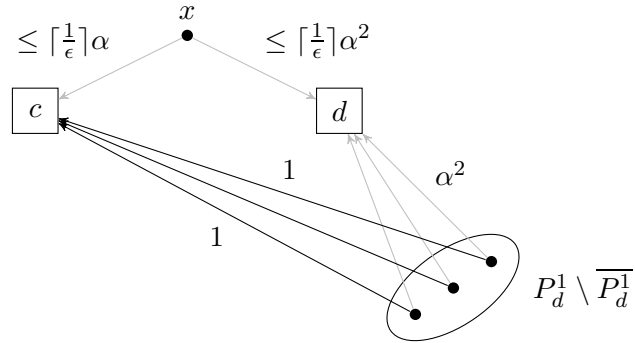
\begin{figure}[t]
	\centering
	\begin{tikzpicture}[dot/.style = {shape = circle, fill = black, inner sep =1.5pt},
	center/.style= {shape = rectangle, draw, minimum size= .6cm},
	every edge/.style={->, draw, >= stealth'},
	every node/.style={color=black},
	light/.style={color=gray!50!white},
	auto]
	\node[dot, label=$x$] (pointx') at (2+9, 1) {};
	\node[dot] (point1') at (5+9, -2.7) {};
	\node[dot] (point2') at (5.5+9, -2.35) {};
	\node[dot] (point3') at (6+9, -2) {};
	\draw[rotate=35] (5.5*0.819+9*0.819-2.35*0.573, -5.5*0.573-9*0.573-2.35*0.819) ellipse (1 and 0.5);
	\node at (7.2+9, -2.35) {$P_d^1\setminus\overline {P_d^1}$};
	\node[center] (centerc') at (0+9,0) {$c$};
	\node[center] (centerd') at (4+9,0) {$d$};
	
	\draw (pointx') edge[light] node[swap]{$\leq \lceil \frac{1}{\epsilon}\rceil \alpha$} (centerc')
	(pointx') edge[light] node{$\leq \lceil \frac{1}{\epsilon}\rceil \alpha^2$} (centerd')
	(point1') edge[light] (centerd')
	(point2') edge[light] (centerd')
	(point3') edge[light] node[swap]{$\alpha^2$} (centerd')
	(point1') edge node{1} (centerc')
	(point2') edge (centerc')
	(point3') edge node[swap]{1} (centerc');
	\end{tikzpicture}
	\caption{Showing the case where $d$ is closed. To bound the distance from points in $P_d^1\backslash \overline {P_d^1}$ to $c$ the respective distances appear with a factor of $\lceil\frac{1}{\epsilon}\rceil\alpha, \lceil\frac{1}{\epsilon}\rceil\alpha^2$ or $\alpha^2$.}
	\label{eBound:CostBound}
\end{figure}

Note that we also prove that we can find a fractional assignment of a special structure. The assignment $\widetilde a$ assigns every point to at most two centers. It is assigned by an amount on one to one center and eventually by an additional amount of $\epsilon$ to a second center. 

\section{A bicriteria algorithm to generalized \kk-median with lower bounds}\label{sec:bicriteria}
So far we presented an algorithm that computes a set of at most $k$ centers $C\subset F$ and an assignment $a\colon P\rightarrow \mathcal P(C)$ such that the lower bound is satisfied at all centers and every point is assigned at least once and at most twice.

An $(\beta,\delta)$-bicriteria solution for generalized $k$-median with lower bounds consists of at most $k$ centers $C'\subset F$ and an assignment $a'\colon P\rightarrow C$ such that at least $\beta B(c)$ points are assigned to $c\in C'$ by $a'$ and $\cost(C',a')\leq \delta \cost(OPT_k)$. Here $OPT_k$ denotes an optimal solution to generalized $k$-median with lower bounds.

Given a $\beta \geq \frac{1}{2}$ and a $\gamma$-approximate solution to generalized $k$-median with 2-weak lower bounds $(C,a)$, we can compute a $(\beta, \gamma\max\{\frac{\alpha\beta}{1-\beta}+1, \frac{\alpha^2\beta}{1-\beta}\})$-bicriteria solution in the following way. Let $C=\{c_1,\ldots, c_{k'}\}$ for some $k'\leq k$. We process the centers in order $c_1,\ldots,c_{k'}$ and decide if they are open or closed. We say that $c_i$ is \textit{smaller} than $c_j$ if $i<j$. If we decide that a center $c$ is open we directly assign at least $\lceil\beta B(c)\rceil$ points to $c$. In the beginning all points are unassigned.

Consider center $c_i$. Let $A_i$ be the set of all points assigned to $c_i$ under $a$. We know that $|A_i|\geq B(c_i)$. If at least $\lceil\beta B(c_i)\rceil$ points in $A_i$ are not assigned so far, $c_i$ remains open and all currently unassigned points from $A_i$ are assigned to $c_i$ (Figure \ref{Bicriteria:Opening}).
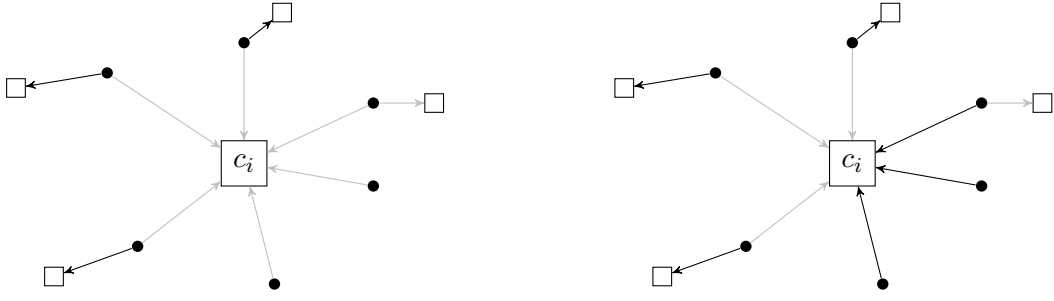
\begin{figure}
	\centering
	\begin{tikzpicture}[dot/.style = {shape = circle, fill = black, inner sep =1.5pt},
	center/.style= {shape = rectangle, draw, minimum size= .6cm},
	emptycenter/.style= {shape = rectangle, draw, minimum size= 7pt, inner sep = 0},
	every edge/.style={->, draw, >= stealth'},
	every node/.style={color=black},
	light/.style={color=gray!50!white},
	auto]
	\node[center] (centerc) at (0,0) {$c_i$};
	\node[emptycenter] (center1) at (-3, 1) {};
	\node[emptycenter] (center2) at (-2.5, -1.5) {};
	\node[emptycenter] (center3) at (0.5, 2) {};
	\node[emptycenter] (center4) at (2.5, 0.8) {};
	
	\node[dot] (dot1) at (-1.8, 1.2) {};
	\node[dot] (dot2) at (-1.4, -1.1) {};
	\node[dot] (dot3) at (0, 1.6) {};
	\node[dot] (dot4) at (1.7, 0.8) {};
	\node[dot] (dot5) at (1.7, -0.3) {};
	\node[dot] (dot6) at (0.4, -1.6) {};
	
	\draw (dot1) edge[light] (centerc)
	(dot2) edge[light] (centerc)
	(dot3) edge[light] (centerc)
	(dot4) edge[light] (centerc)
	(dot5) edge[light] (centerc)
	(dot6) edge[light] (centerc)
	(dot1) edge (center1)
	(dot2) edge (center2)
	(dot3) edge (center3)
	(dot4) edge[light] (center4);
	
	\node[center] (centerc') at (0+8,0) {$c_i$};
	\node[emptycenter] (center1') at (-3+8, 1) {};
	\node[emptycenter] (center2') at (-2.5+8, -1.5) {};
	\node[emptycenter] (center3') at (0.5+8, 2) {};
	\node[emptycenter] (center4') at (2.5+8, 0.8) {};
	
	\node[dot] (dot1') at (-1.8+8, 1.2) {};
	\node[dot] (dot2') at (-1.4+8, -1.1) {};
	\node[dot] (dot3') at (0+8, 1.6) {};
	\node[dot] (dot4') at (1.7+8, 0.8) {};
	\node[dot] (dot5') at (1.7+8, -0.3) {};
	\node[dot] (dot6') at (0.4+8, -1.6) {};
	
	\draw (dot1') edge[light] (centerc')
	(dot2') edge[light] (centerc')
	(dot3') edge[light] (centerc')
	(dot4') edge (centerc')
	(dot5') edge (centerc')
	(dot6') edge (centerc')
	(dot1') edge (center1')
	(dot2') edge (center2')
	(dot3') edge (center3')
	(dot4') edge[light] (center4');
	\end{tikzpicture}
	\caption{Shows the case where $A_i$ contains at least $\lceil\beta B(c_i)\rceil$ unassigned points. The three points on the left are already assigned to other centers and the three points on the right are newly assigned to $c_i$. The gray connections come from $a$.}
	\label{Bicriteria:Opening}
\end{figure}
If less than $\lceil\beta B(c_i)\rceil$ points from $A_i$ are unassigned, the center is closed.

Let $C'$ denote the centers from $\{c_1,\ldots,c_{i-1}\}$ which are open and $B_i$ the set of unassigned points from $A_i$ which are not connected to any center larger than $c_i$ under $a$. To guarantee that all points are assigned at the end, we have to care about points in $B_i$. By assumption there are at most $\lfloor\beta B(c_i)\rfloor$ such points. We simply assign point $p\in B_i$ to the nearest center $\argmin_{c\in C'}d(c,p)$ in $C'$. The whole procedure is described in Algorithm~\ref{algo:bicriteriell}.

To upper bound the assignment cost in the case $c_i$ is closed by the algorithm we consider a second assignment $b$, which may be fractional. We define for $p\in B_i$ and $c\in C'$ a value $b_p^c\in [0,1]$ which indicates the amount by which $p$ is assigned to $c$. We claim that we can find a fractional assignment such that for every $q\in B_i$ and $f\in C'$ the following holds
\begin{enumerate}
	\item point $q$ is assigned by an amount of one, i.e., 
	\[\sum_{c\in C'}b_q^c=1\]
	\item and at most $\frac{\beta}{1-\beta}|\{p\in A_i \mid f\in a(p)\}|$ amount is assigned to $f$, i.e., 
	\[\sum_{p\in B_i}b_f^p\leq \frac{\beta}{1-\beta}|\{p\in A_i \mid f\in a(p)\}|\]
\end{enumerate}
Such an assignment can be found since
\begin{align*}
	\frac{\beta}{1-\beta} \sum_{c\in C'} |\{p\in A_i \mid c\in a(p)\}|=&\frac{\beta}{1-\beta}|\{p\in A_i\mid a(p)\cap C'\neq \emptyset\}|\\
	\geq &\frac{\beta(1-\beta)}{1-\beta} B(c_i)\geq |B_i|.
\end{align*}
To see the first inequality we observe the following.
If a point $p\in A_i$ is connected to an open center $c\in C'$ under $a$, it is already assigned to $c$ by the algorithm. So the set of points from $A_i$ which are already assigned to some center equals $\{p\in A_i\mid a(p)\cap C'\neq \emptyset\}$. We know that $|A_i|\geq B(c_i)$ and that at most $\lfloor\beta B(c_i)\rfloor$ points from $A_i$ are unassigned. Thus we have $|\{p\in A_i\mid a(p)\cap C'\neq \emptyset\}|\geq (1-\beta)B(c_i)$.

Let $b$ be an assignment satisfying the above properties. We obtain the following upper bound to the cost of $b$.
\begin{align*}
	\sum_{c\in C'}\sum_{p\in B_i}  b_p^c\ d(p,c)
	\leq& \sum_{c\in C'}\Big(\sum_{x\in A_i\colon c\in a(x)}\frac{\beta}{1-\beta} \big ( \alpha^2 d(c_i,x)+\alpha d(x,c)\big )+\sum_{p\in B_i}b_p^c \alpha^2 d(p,c_i)\Big )\\
	=&\frac{\alpha\beta}{1-\beta} \sum_{c\in C'} \sum_{x\in A_i\colon c\in a(x)} d(x,c)+ 
	\frac{\alpha^2\beta}{1-\beta} \sum_{x\in A_i} d(x,c_i).
\end{align*}
For the first inequality we used above bound on $\sum_{p\in B_i}b_p^c$. We can charge every point in $\{x\in A_i\mid c\in a(x)\}$ up to an amount of $\frac{\beta}{1-\beta}$ for the assignment cost of $B_i$ to $c$. Assume such a point $x$ gets charged by an amount of $\gamma\leq b_p^c$ for the distance $d(p,c)$. We obtain the following upper bound on the cost
\[\gamma d(p,c)\leq \gamma(\alpha^2 d(p,c_i)+\alpha^2 d(c_i,x)+\alpha d(x,c)).\] 
Thus in total the distance $d(p,c_i)$ appears with a factor of $b_p^c\alpha^2$, distance $d(c_i,x)$ with factor $\frac{\beta}{1-\beta} \alpha^2$ and $d(x,c)$ with factor $\frac{\beta}{1-\beta} \alpha$ in the upper bound on the assignment cost of $B_i$ to $c$. 

The equality follows immediately from $\sum_{c\in C'}b_p^c=1$ and $B_i\cap \{x\in A_i\mid a(x)\cap C'\neq\emptyset\}=\emptyset$.

Assigning every point in $B_i$ to its nearest center can only be cheaper than distributing $B_i$ to centers in $C'$ via $b$. We obtain 
\begin{align}
	\nonumber\sum_{p\in B_i} \min_{c\in C'} d(p,c)\leq& \sum_{c\in C'}\sum_{p\in B_i}  b_p^c\ d(p,c)\\
	\leq &\frac{\alpha\beta}{1-\beta} \sum_{c\in C'} \sum_{x\in A_i\colon c\in a(x)} d(x,c)+ 
	\frac{\alpha^2\beta}{1-\beta} \sum_{x\in A_i} d(x,c_i).\label{bik:ineq}
\end{align}

\begin{figure}
	\centering
	\begin{tikzpicture}[dot/.style = {shape = circle, fill = black, inner sep =1.5pt},
	center/.style= {shape = rectangle, draw, minimum size= .6cm},
	emptycenter/.style= {shape = rectangle, draw, minimum size= 10pt, inner sep = 0},
	every edge/.style={->, draw, >= stealth'},
	every node/.style={color=black},
	light/.style={color=gray!50!white},
	auto]
	\node[center] (centerc) at (0,0) {$c_i$};
	\node[emptycenter] (center1) at (-3, 1) {};
	\node[emptycenter] (center2) at (-2.5, -1.5) {};
	\node[emptycenter] (center3) at (0.5, 2) {};
	\node[emptycenter] (center4) at (2.5, 0.8) {};
	
	\node[dot] (dot1) at (-1.8, 1.2) {};
	\node[dot] (dot2) at (-1.4, -1.1) {};
	\node[dot] (dot3) at (0, 1.6) {};
	\node[dot] (dot4) at (1.7, 0.8) {};
	\node[dot] (dot5) at (1.7, -0.3) {};
	\node[dot] (dot6) at (0.4, -1.6) {};
	\coordinate (help11) at (-0.15, -0.45);
	\coordinate (help12) at (-1.3, -1.28);
	\coordinate (help21) at (0.45, 0.15);
	\coordinate (help22) at (0.45, 0.45);
	\coordinate (help23) at (-1.8, 1.4);
	
	\draw (dot1) edge[light] (centerc)
	(dot2) edge[light] node{$\alpha^2$} (centerc)
	(dot3) edge[light] (centerc)
	(dot4) edge[light] (centerc)
	(dot5) edge[light] (centerc)
	(dot6) edge[light] node[swap, near start]{$\alpha^2$} (centerc)
	(dot1) edge (center1)
	(dot2) edge node[swap]{$\alpha$} (center2)
	(dot3) edge (center3)
	(dot4) edge[light] (center4);
	\draw [->, >= stealth', rounded corners] (dot6) to (help11) to (help12) to (center2);
	\draw [->, >= stealth', rounded corners] (dot5) to[bend right=3] (help23) to (center1);
	\end{tikzpicture}
	\caption{Showing assignment $b$ in the case where $c_i$ is closed. The two points from $B_i$ are distributed to centers in $C'$. The gray connections come from $a$. $\alpha$ and $\alpha^2$ are the factors with which the respective distances appear in the upper bound of the new connection. }
	\label{Bicriteria:Closing}
\end{figure}
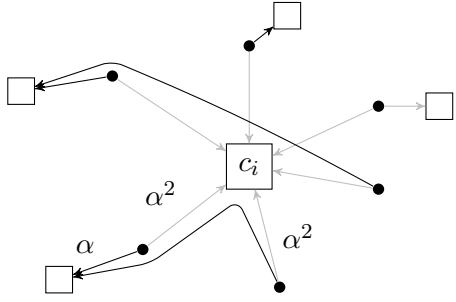

Let $(C',a')$ be the final solution computed by the algorithm.
\begin{align*}
	\cost	&(C',a')=\sum_{c\in C'}\sum_{\substack{x\in P:\\ a'(x)=c}}d(x,c)\\
	&\leq \sum_{c\in C'}\hspace{-2.5pt}\sum_{\substack{x\in P:\\ c\in a(x)}} d(x,c)+\frac{\alpha\beta}{1-\beta}\sum_{c\in C'}\hspace{-2.5pt}\sum_{\substack{x\in P:\\ c\in a(x)}} d(x,c)+\frac{\alpha^2\beta}{1-\beta}\sum_{c\in C\backslash C'}\sum_{\substack{x\in P:\\c\in a(x)}} d(x,c)\\
	&\leq \max\{\frac{\alpha\beta}{1-\beta}+1, \frac{\alpha^2\beta}{1-\beta}\}\sum_{c\in C}\sum_{\substack{x\in P:\\ c\in a(x)}} d(x,c)\\
	&=\max\{\frac{\alpha\beta}{1-\beta}+1, \frac{\alpha^2\beta}{1-\beta}\}\cost(C,a).
\end{align*}
To see the first inequality we use the upper bound in \eqref{bik:ineq}. Let $x\in P$ and $c\in a(x)$. If $c$ is closed in the final solution the distance $d(x,c)$ is only charged with a factor of $\frac{\alpha^2\beta}{1-\beta}$ in \eqref{bik:ineq} for closing $c$. If $c$ is open in the final solution the distance $d(x,c)$ is charged with factor one if $a'(x)=c$ and can also be charged with a factor of $\frac{\alpha\beta}{1-\beta}$ in \eqref{bik:ineq} for closing a center $d\in a(x)$. This can happen at most once since $|a(x)|\leq 2$. This proves the first inequality.

Since generalized $k$-median with 2-weak lower bounds is a relaxation of generalized $k$-median with lower bounds we obtain 
\[\cost(C',a')\leq \gamma \max\{\frac{\alpha\beta}{1-\beta}+1, \frac{\alpha^2\beta}{1-\beta}\}\cost(OPT).\]
This leads to the following theorem.
\IncMargin{1ex}
\begin{algorithm}[t]
	\DontPrintSemicolon
	\SetKwInOut{Input}{Input}\SetKwInOut{Output}{Output}
	\SetKwFor{ForAll}{for all}{}{}
	\Input{$\gamma$-approximate solution $(C,a)$ to  generalized $k$-median with 2-weak lower bounds, $C=\{c_1,\ldots, c_{k'}\}$}
	\Output{Bicriteria solution $(C',a')$ to generalized $k$-median with lower bounds.}
	\BlankLine
	set $C'=\emptyset$, $a'(x)=\bot$ for all $x\in P$\;
	$N=P$\;
	\For{$i=1$ \KwTo $k'$}
	{	
		$A_i=\{x\in P\mid c_i\in a(x)\}$\;
		$B_i=\{x\in A_i\mid a(x)\subset \{c_1\ldots,c_i\}\}\cap N$\;
		\If{$A_i\cap N\geq \beta B(c_i)$}
		{	
			set $a'(x)=c_i$ for all $x\in A_i\cap N$\;
			$N=N\backslash A_i$\;
			$C'=C'\cup \{c_i\}$\;
			
		}
		\Else{
			set $a'(x)=\argmin_{c\in C'}d(x,c)$ for all $x\in B_i$\;
		}
		
	}
	\caption{A $(\beta, \gamma \max\{\frac{\alpha\beta}{1-\beta}+1, \frac{\alpha^2\beta}{1-\beta}\})$-bicriteria approximation algorithm to generalized $k$-median with lower bounds\label{algo:bicriteriell}}
\end{algorithm}
\DecMargin{1ex}   

\thmBicriteria*

\end{document}